\documentclass[11pt]{article}

\usepackage{amsfonts,amsmath,amssymb}
\usepackage{multirow,ae,aecompl,url} 
\usepackage{graphicx,xcolor}
\usepackage{float}
\usepackage[inline]{enumitem}

\newcommand{\FF}{\vspace*{\medskipamount}}

\newcommand{\BBB}{\vspace*{-\bigskipamount}}

\newcommand{\cE}{\mathcal{E}}

\newcommand{\cG}{\mathcal{G}}

\newcommand{\cO}{\mathcal{O}}

\newcommand{\Paragraph}[1]{\BBB\paragraph{#1}}
\newcommand{\remove}[1]{}

\newcommand{\ceil}[1]{\left\lceil #1 \right\rceil}
\newcommand{\polylog}{\text{polylog }}

\newcommand{\dk}[1]{{\color{black}#1}}

\newcommand{\qed}{\hfill $\square$ \smallbreak}
\newenvironment{proof}{\noindent{\bf Proof:}}{\qed}

\newtheorem{theorem}{Theorem}
\newtheorem{lemma}{Lemma}

\newtheorem{corollary}{Corollary}
\newtheorem{proposition}{Proposition}

\begin{document}

\title{Deterministic Fault-Tolerant Distributed Computing\\ 
in Linear Time and Communication \vfill}

\author{	Bogdan S. Chlebus \footnotemark[1]
		\and
		Dariusz R. Kowalski \footnotemark[2]
		\and
		Jan Olkowski \footnotemark[3]}  	

\footnotetext[1]{	School of Computer and Cyber Sciences, Augusta University, Augusta, Georgia, USA.}

\footnotetext[2]{	School of Computer and Cyber Sciences, Augusta University, Augusta, Georgia, USA. Partially supported by 
the NSF grant 2131538.}

\footnotetext[3]{	Department of Computer Science, University of Maryland, College Park, Maryland, USA.}

\date{}

\maketitle

\vfill

\begin{abstract}
We develop deterministic algorithms for the problems of consensus, gossiping and checkpointing with nodes prone to failing.
Distributed systems are modeled as synchronous complete networks. 
Failures are represented either as crashes or authenticated Byzantine faults.
The algorithmic goal is to have  both linear running time and linear amount of communication for as large an upper bound $t$ on the number of faults as possible, with respect to the number of nodes~$n$.
For crash failures, these bounds of optimality are $t=\cO(\frac{n}{\log n})$ for consensus and $t=\cO(\frac{n}{\log^2 n})$ for gossiping and checkpointing, while the running time for each algorithm is $\Theta(t+\log n)$.
For the authenticated Byzantine model of failures, we show how to accomplish both linear running time and communication for  $t=\cO(\sqrt{n})$.
We show how to implement the algorithms in the single-port model, in which a node may choose only one other node to send/receive a message to/from in a round, such as to preserve the range of running time and communication optimality.
We prove lower bounds to show the optimality of some performance bounds. 

\vfill

\noindent
\textbf{Keywords:} 
crash failure, 
Byzantine fault, 
authentication,
consensus, 
checkpointing, 
gossiping,
Ramanujan graph,
lower bound
\end{abstract}

\vfill

~

\thispagestyle{empty}

\setcounter{page}{0}


\newpage

\section{Introduction}

\label{sec:introduction}

We consider deterministic algorithms in synchronous message-passing distributed systems. 
Nodes are prone to failures, which are modeled either as crashes or authenticated Byzantine faults.
Running time is a natural performance metric of algorithms in synchronous systems.
The amount of communication is another performance metric, for which we use either the number of point-to-point messages or the total number of bits in these messages.
A node can send a direct message to any other node in one round. 
To structure communication, we use overlay networks with topologies conducive to support efficient exchange of information.

Let $n$ be the number of nodes and $t$ be an upper bound on the number of faults in an execution.
Both numbers $n$ and $t$ are known, in that they can be used in codes of algorithms. 
Natural algorithmic problems in the considered settings require both linear time $\Omega(t)$ and linear number of messages $\Omega(n)$, even with at most one faulty node, to be solved deterministically, which includes consensus, gossiping and checkpointing.
A deterministic distributed algorithm solving such problems has ultimate optimal performance if it works within time~$\cO(t)$ and with communication~$\cO(n)$. 
For numbers $0<t<n$, the closer~$t$ is  to~$n$, the more challenging it is for an algorithm to attain efficiency with respect to both running time and communication.
In this work, we aim to develop deterministic distributed algorithms for consensus, gossiping and checkpointing that are asymptotically optimal with respect to both running time and communication for bounds  on the number of faults~$t$ that are as close to the total number~$n$ of nodes as possible.

The considered problems have similar goals to achieve, which makes it possible to use related  paradigms in design of their algorithmic solutions. 
Each node is initialized with an input and it eventually produces an output. 
If the output is the same at some nodes then these nodes \emph{agree} on this output.
For consensus and checkpointing, all non-faulty nodes need to eventually agree on the same output.
For checkpointing, the output is a set of nodes that includes each node operational at the round of termination and excludes every node that crashed at the start of an execution.
Gossiping is a weaker version of checkpointing, in which nodes do not need to agree on the outputs.

We explore properties of Ramanujan graphs that are related to fault-tolerance and quantitative analysis of proliferation of information in networks with topology modeled by such graphs. 
We show in Section~\ref{sec:ramanujan-graphs} that Ramanujan graphs have specific topological properties conducive to fault-tolerance and time/communication efficiency of distributed algorithms.
We present deterministic algorithms for agreement-related problems in Section~\ref{sec:agreement-crashes}.
The problems include  almost-everywhere agreement, spreading common value, and consensus.
Nodes are initialized with binary values $0$ or $1$ and messages that nodes exchange carry only one bit. 

Next, we review the contributions.
Faults are modeled as crashes, unless stated otherwise.
We give an algorithm for almost-everywhere agreement in which at least $\frac{3}{5} n$ nodes decide on the same value, assuming $t<\frac{n}{5}$.
We give an algorithm for spreading common value, in which initially at least $\frac{3}{5} n$ nodes are initialized with a common value and eventually every node decides on this value, assuming $t<\frac{n}{5}$.
We give an algorithm for consensus has the property that an execution takes $\cO(t+\log t)$ rounds and $\cO(n+t \log t)$ bits  are transmitted in total, assuming $t<\frac{n}{5}$.
The efficiency of the algorithm is reflected by the properties that one crash delays termination by $\cO(1)$ rounds, there are $\cO(1)$ bits transmitted per node, and $\cO(\log t)$ bits transmitted per crash.
The optimum number of bits/messages $\cO(n)$ are sent as long as $t=\cO(\frac{n}{\log n})$.
Optimizing the bit-communication performance of algorithms for binary consensus with nodes prone to crashes crashes was addressed by Galil, Mayer, and Yung~\cite{GalilMY95}, who developed a consensus algorithm that sends $\cO(n)$ bits in messages for any bound $t<n$ on the number of crashes, but it runs in time that is exponential in~$n$. 
We also give an algorithm for consensus for arbitrary $0<t<n$, which works in  at most $n+3(1+\lg n)$ rounds and  nodes send at most $\bigr(\frac{5}{1-\alpha}\bigr)^{8} \,n\lg n$ one-bit  messages, where $\alpha=\frac{t}{n}$.

We present a gossiping algorithm working in $\cO(\log n \log t)$ rounds while nodes send $\cO(n+t\log n \log t)$ messages (of linear size) in Section~\ref{sec:gossip}.  
We give next a checkpointing algorithm in Section~\ref{sec:checkpoining}, which is based on our consensus and gossiping algorithms.
The algorithm for checkpointing works in linear time $\cO(t)$ and nodes send $\cO(n + t\log n \log t)$ messages. 
This improves on the  most message-efficient time-optimal solution previously known, by Galil, Mayer,  and Yung~\cite{GalilMY95}, by a  polynomial factor.


The model of communication we assume is multi-port, unless stated otherwise.
We show that our consensus, gossiping and checkpointing algorithms can be implemented in the single-port model with similar asymptotic running times and the same communication performance bounds as in the multi-port model. 
We also prove  in Section~\ref{sec:single-port} that running time $\Omega(t+\log n)$ in necessary for all the three considered problems in the single-port model.

Table~\ref{tab:summary-results} summarizes the ranges for the number of faults $t<n$ for which the algorithms attain the best possible linear time $\cO(t)$ and linear~communication~$\cO(n)$.
These performance bounds hold for the multi-port setting, but some of them also hold in the single-port model, in which case the optimum time becomes $\Theta(t+\log n)$, by the lower bound given in Section~\ref{sec:single-port}.

\begin{table}
\label{tab:summary-results}
\centering
 \begin{tabular}{|c|c|c|c|c|}
\hline
fault type & problem & optimality range of $t$ & reference & single port?  \\
\hline \hline
& consensus & $\cO(1)$ & \cite{GalilMY95} & N/A \\
crash&consensus & $\cO(\frac{n}{\log n})$  & Sec.~\ref{sec:agreement-crashes} & Yes\\
failures  &gossip/checkpointing & $\cO(1)$  & \cite{GalilMY95} & N/A \\
&gossip/checkpointing & $\cO\bigl(\frac{n}{\log^2 n}\bigr)$  & Sec.~\ref{sec:gossip}--\ref{sec:checkpoining} &  Yes\\
\hline
 authenticated & consensus & $\cO(1)$ & \cite{DolevS83} & No\\
Byzantine & consensus & $\cO\left(\sqrt{n}\right)$ & Sec.~\ref{sec:ab-consensus} & No \\
 \hline
\end{tabular}

\caption{\label{table:summary}
The ranges for a bound $t$ on the number of crashes for which a deterministic solution has both optimal  running time~$\cO(t)$ and  communication~$\cO(n)$. 
The communication cost denotes the number of bits in messages in the case of consensus, and it is the number of messages for gossiping and checkpointing.
For authenticated Byzantine faults, we consider only messages sent by non-faulty nodes.
Nodes send messages in the multi-port mode.
The single-port question is about the possibility to adapt the algorithm to single port while maintaining asymptotic performance bounds.
The mark N/A means that the previous work did not consider the single-port model, and  it is not apparent how to implement the algorithm in that model without increasing asymptotic performance. }
\end{table}

\Paragraph{Previous work.}

A consensus algorithm needs to send~$\Omega(n)$ messages because each node is required to send at least one message.
Galil, Mayer, and Yung~\cite{GalilMY95} developed an algorithm that has $\cO(n)$ messages sent and which runs in $\cO(n^{1+\varepsilon})$ rounds, for any $0<\varepsilon<1$.
They also gave an algorithm for binary consensus sending $\cO(n)$ bits in messages, but the algorithm runs in exponential time.
Chlebus and Kowalski~\cite{ChlebusK-JCSS06} showed that consensus can be solved  by a deterministic algorithm in  $\cO(t)$ time and with $\cO(n\log^2 t)$ messages with the assumption that the number $n-t$ of non-faulty nodes satisfies $n-t=\Omega(n)$.
Chlebus, Kowalski, and Strojnowski~\cite{ChlebusKS-PODC09} gave deterministic algorithms for binary consensus operating in time $\cO(t)$ that send $\cO(n\log^2 n)$ bits for $t<\frac{n}{3}$ and $\cO(n\log^4 n)$ bits for any $t<n$.

The problem of checkpointing was first considered by De Prisco, Mayer, and Yung~\cite{DePriscoMY94}, who gave an algorithm with $\cO(tn)$ message performance.
A checkpointing algorithm developed by Chlebus, Gasieniec, Kowalski, and Schwarzmann~\cite{ChlebusGKS-IC2017} sends $\cO(n(n-t))$ messages.
Galil, Mayer, and Yung~\cite{GalilMY95} gave an algorithm solving checkpointing in time $\cO(t\,8^{1/\epsilon})$ and with $\cO(n+t\, n^{\epsilon})$ messages, for any $\epsilon>0$. 
All the previously known algorithms for consensus and checkpointing are optimal with respect to both running time and communication performance for~$t=\cO(1)$.

Fault-tolerant gossiping was introduced by Chlebus and Kowalski~\cite{ChlebusK-JCSS06}.
They developed a deterministic algorithm solving gossiping with running time $\cO(\log^2 t)$ while generating $\cO(n\log^2 t)$ messages, provided $n-t=\Omega(n)$. 
They also showed a lower bound  $\Omega(\frac{\log n}{\log(n\log n)-\log t})$ on the number of rounds in case $\cO(\polylog n)$ amortized messages are used per node.
Algorithms attaining $\cO(n\,\polylog n)$ message performance while maintaining poly-logarithmic running time for any $t<n$ are known in the literature, see~\cite{ChlebusK-JCSS06, ChlebusK-DISC06,georgiou2005efficient}.
Algorithms  for the studied problems of consensus, gossiping and checkpointing that are efficient with respect to both running time and communication when implemented in the single-port model have not been known prior to this work, to the best of the authors'  knowledge.


\Paragraph{Related work.}

Consensus is among the central problems in distributed  and communication algorithms.
It was introduced by Pease, Shostak, and Lamport~\cite{PeaseSL80} and Lamport, Shostak,  and Pease~\cite{LamportSP82}.
A consensus algorithm is early-stopping  if it runs in $\cO(f+1)$ time, where~$f$ is the number of crashes  actually occurring in an execution, while the algorithm may be designed for a known upper bound~$t<n$ on the number of crashes.
Dolev, Reischuk, and Strong~\cite{DolevRS90} gave an early-stopping solution for consensus with an arbitrary bound $t<n$ on the number of node crashes  and a lower bound $\min\{ t+1,f+2\}$ on the running time. 
Coan~\cite{Coan93} gave a consensus algorithm running in time $\cO(f+1)$ that uses messages of size logarithmic in the size of the range of input values; see also Bar-Noy, Dolev, Dwork, and Strong~\cite{Bar-NoyDDS92} and Berman, Garay, and Perry~\cite{BermanGP92}.
Chlebus and Kowalski~\cite{ChlebusK-DISC06} developed an early stopping consensus algorithm sending $\cO(n \log^5 n)$ messages. 
Dolev and Lenzen~\cite{DolevL13} showed that any crash-resilient  consensus  algorithm deciding in exactly $f + 1$ rounds has  $\Omega(n^2f)$ worst-case message complexity.
Randomized approach has been considered by Chor, Merritt, and Shmoys~\cite{ChorMS89} who gave a randomized algorithm for consensus that has $\cO(\log n)$ round complexity and $\cO(n^2\log n)$ message complexity with high probability, while tolerating fewer than~$\frac{n}{2}$ crashes.
Bar-Joseph and Ben-Or~\cite{Bar-JosephB98} gave a randomized consensus algorithm against an adaptive adversary that controls crashes.
The algorithm works in $\cO\bigl(\frac{\sqrt{n}}{\log{n}}\bigr)$ expected time, which is provably optimal, while generating $\cO\bigl(\frac{n^{5/2}}{\log{n}}\bigr)$ messages and communications bits.
Kowalski and Mirek~\cite{KowalskiM19} demonstrated how to decrease the number of messages to  $\cO(n^{3/2}\,\text{ polylog } n)$, while keeping the number of bits at $\Theta(n^{5/2}\text{\,polylog } n)$ and slowing down the algorithm by a factor of $\cO(\log^2 n)$, by using deterministic fault-tolerant gossip developed in~\cite{ChlebusK-DISC06}.
Chlebus and Kowalski~\cite{ChlebusK-SPAA09} developed a randomized  consensus algorithm that terminates in the expected $\cO(\log n)$ time and such that the expected number of bits sent and received by each node is~$\cO(\log n)$ when the adversary is  oblivious and such that a bound~$t$ on the number of crashes is a constant fraction of the number of nodes.  
Gilbert and Kowalski~\cite{GilbertK10} presented a randomized consensus algorithm that tolerates up to~$\frac{n}{2}$ crashes and  terminates in $\cO(\log n)$ time and sends $\cO(n)$ messages with high probability.
Gilbert, Guerraoui, and Kowalski~\cite{GilbertGK07} developed an indulgent  consensus algorithm, that solves consensus assuming eventual synchrony, while in synchronous executions it is early-stopping and achieves  $\cO(n\,\text{polylog }n)$ message performance.
Robinson, Scheideler, and Setzer~\cite{RobinsonSS18} showed how to achieve an almost-everywhere consensus  in $\cO(\log n)$ time with high probability against adversaries controlling crashes that are weaker than adaptive.
Consensus has been also considered in settings beyond crashes in synchronous systems. 
Dolev and Reischuk~\cite{DolevR85} and Hadzilacos and Halpern~\cite{HadzilacosH93} proved the $\Omega(tn)$ lower bound on the message performance of deterministic consensus for authenticated Byzantine failures. 
King and Saia~\cite{KingS11} showed that subject to some limitation on the adversary and requiring termination with high probability, a sub-quadratic expected communication $O(n^{3/2}\polylog{n})$ can be achieved even for general Byzantine faults. 
Abraham et al.~\cite{AbrahamCDNP0S19} showed the necessity of such limitations to achieve sub-quadratic running time  for Byzantine faults.
In asynchronous settings, Alistarh, Aspnes, King and Saia~\cite{AlistarhAKS18} showed how to obtain almost optimal communication complexity $\cO(n t + t^2 \log^2 t)$ if fewer than $\frac{n}{2}$ nodes may fail, which improved upon the previous result $\cO(n^{2}\log^2 {n})$ by Aspnes and Waarts~\cite{AspnesW96} and is asymptotically almost optimal due to a lower bound shown by Attiya and Censor-Hillel~\cite{AttiyaH10}.
Chlebus, Kowalski, and Strojnowski~\cite{ChlebusKS-DISC10} gave  a quantum algorithm for binary consensus, executed by crash-prone quantum processors, that operates in $\cO(\text{polylog }n)$ rounds while  sending $\cO(n \text{ polylog } n)$ qubits against the adaptive adversary.

\Paragraph{General perspective.}

Informative expositions of formal models of distributed and networked systems prone to failures of components, as frameworks for development and study of distributed algorithms,  include the books by Attiya and Welch~\cite{Attiya-Welch-book2004}, Herlihy and Shavit~\cite{HerlihyShavit-book}, Lynch~\cite{Lynch-book96}, and Raynal~\cite{Raynal2010-Synthesis}.
Attiya and Ellen~\cite{Attiya-Ellen-book-2014} review techniques to show  impossibilities and lower bounds for algorithmic problems in distributed and networked systems. 
For studies of properties and construction of graphs with suitable expansion properties, see~\cite{DavidoffSV-book03,HooryLW06,KrebsS2011-book}.

\section{Technical Preliminaries}
\label{sec:preliminaries}

We model distributed systems as synchronous networks consisting of $n$ nodes.
Each node has a unique integer name in the set $[n]=\{1,\ldots,n\}$.
Nodes are prone to failing, which is modeled either as crashes or authenticated Byzantine faults.
The letter $t$ denotes an upper bound on a possible number of faults.
The numbers~$n$ and $t$ are known to every node and can be parts of codes  of algorithms.
A crashed node stops any activity.
A node faulty in the authenticated Byzantine sense may undergo arbitrary state transitions but it cannot forge messages claiming that they are forwarded from other nodes. 

Timings of specific crashes are determined by an adversary who is constrained by an upper bound~$t$ on the total number of nodes that may crash in an execution of an algorithm.
A node that crashes at a round stops any activity in the following rounds, and does not send nor receive messages.
A node that has not crashed by a round is \emph{operational} at the round.
We disregard crash of a nodes after the node halts in an execution while operational, which is the same as assuming that a node does not crash after halting voluntarily.
Nodes that halt voluntarily while operational are \emph{non-faulty} in the execution.

The nodes communicate among themselves by sending messages.
Any pair of nodes can directly exchange messages at every round.
During a round, all messages sent to a node in this round get delivered.
If a node can multicast and receive messages to/from any set of recipients/senders at a round then this is the \emph{multi-port model}.
If a node can send/receive a message to/from at most one sender/recipient at a round then we call this the \emph{single-port model}.
A node does not obtain any signal from any of its ports that messages have been delivered to the port and need to be received by the node.


We use the running time and amount of communication as performance metrics  of algorithms.
Nodes begin execution of an algorithm at the same round.
\emph{Runtime} performance is measured by the number of rounds that occur until all non-faulty nodes have halted.
\emph{Communication} performance is measured either by the number of messages or the total number of bits in all messages sent  in the course of an execution; \dk{in case of Byzantine faults, we count only messages sent by non-faulty nodes (as Byzantine nodes could flood the system with arbitrary number of messages)}.


\Paragraph{Algorithmic problems.}

We develop algorithms for problems related to reaching agreement.
Each of these problems is parametrized by a number of nodes $n$ and an upper bound on the number of crashes~$t$.
We want algorithms to have a property that every node eventually halts, unless it fails.
A node \emph{decides} on a value if it sets a dedicated variable to the value, and this assignment of value to a variable is irrevocable once made.

In the problems of consensus, each node starts with its initial input value. 
We assume additionally that each initial value  is either $0$ or~$1$.
Algorithms solving consensus need to satisfy the requirements of validity, agreement, and termination. 
\emph{Validity} means that a decision is on the input value of some node.  
\emph{Agreement} means that no two nodes decide on different values.
\emph{Termination} means that each node eventually decides, unless it fails.

The problem of almost everywhere agreement is about a fraction of nodes agreeing on some common value.
The $\kappa$-almost-everywhere-agreement, or $\kappa$-AEA for short, has a number parameter~$\kappa$ such that $\frac{1}{2}<\kappa<1$.
The algorithmic task to accomplish is to have at least $\kappa n$ nodes that  eventually either decide or fail.
The agreement and validity requirements apply only to nodes that actually decide.

A problem we call $\kappa$-spread-common-value, or $\kappa$-SCV for short, has a number parameter $\kappa$ such that $\frac{1}{2}<\kappa<1$.
An instance of the problem is determined by  initializing a dedicated variable  in at least $\kappa n$ nodes to  a \emph{common value} that is not \texttt{null}, and initializing this variable to \texttt{null} at the remaining nodes.
The algorithmic task in this problem is to have every non-faulty node eventually decide on the common value.

In the problems of checkpointing and gossiping,  each node works to create a collection of node names and their associated values, called an \emph{extant set}, and eventually decide on it, unless crashing.
An extant set of a node that decided on it is a \emph{decided} extant set.

In checkpointing, an extant set is just a collection of nodes' names without any associated values.
We want to accomplish deciding on extant sets subject to the following conditions: 
(1)~if a node~$j$ crashed prior to sending any message then $j$ is not in any decided  extant set, (2)~if a node~$j$  halted operational then~$j$ belongs to all decided extant sets, and (3)~all decided extant sets are equal. 

In gossiping, each node is initially equipped with an input value called its \emph{rumor}.
An extant set consists of pairs made of a node's name and its rumor; a pair consisting of a node's $i$ name and its rumor is the \emph{node~$i$'s pair}.
We want to accomplish deciding on extant sets subject to the following conditions:
(1)~if a node~$j$ crashed prior to sending any message then $j$'s pair is not in any  decided extant set, (2)~if a node~$j$ halted operational then $j$'s pair  belongs to all decided extant sets.
We do not require that all decided extant sets are equal in the case of gossiping.

\Paragraph{Overlay graphs.}

Algorithms use overlay graphs of choice to save on communication.
We consider nodes as vertices and messages are sent only between specific pairs of nodes connected by edges in the graph.
Let $G=(V,E)$ denote a simple graph, where $V$ is the set of vertices and $E$ is the set of edges. 
The \emph{subgraph of~$G$ induced by~$W\subseteq V$}, denoted~$G|_W$, is the subgraph of~$G$ containing the vertices in~$W$ and all the edges with both endpoints in~$W$.
A node adjacent to a node~$v$ is a \emph{neighbor of~$v$} and the set of all the neighbors of a node~$v$ is the \emph{neighborhood of~$v$}.
The notation $N^i_G(W)$  denotes the \emph{generalized neighborhood of~$W$ of radius~$i$}, which is the set of all vertices in~$V$ of distance at most~$i$ from some node in~$W$  in graph~$G$.
In particular, the neighborhood of~$v$ is denoted~$N_G(v)=N^1_G(v)$.
For two disjoint set of vertices $W_1$ and $W_2$, an edge $(v,w)\in E$  \emph{connects set~$W_1$ with~$W_2$}  if $v\in W_1$ and $w\in W_2$. 
The number of edges connecting $W_1$ with $W_2$ is denoted by~$e(W_1,W_2)$.
Next, we list properties of overlay graphs and their vertices relevant to efficiency of algorithms, following~\cite{ChlebusKS-PODC09}.
Let $\delta$, $\gamma$ and~$\ell$ be positive integers and $0< \varepsilon <1$ be a real number.
\begin{description}
\item[\rm\em Dense neighborhood:]
For a node $v\in V$, a set $S\subseteq N^{\gamma}_G(v)$  is said to be \emph{$(\gamma,\delta)$-dense-neighborhood for~$v$} if every node in~$S\cap N^{\gamma-1}_G(v)$ has at least $\delta$ neighbors in~$S$.

\item[\rm\em Survival subset:]
For a set of vertices $B\subseteq V$, a subset $C\subseteq B$ is a \emph{$\delta$-survival subset for~$B$} if every node's degree in~$G|_C$ is at least~$\delta$. 

\item[\rm\em Compactness:]
graph~$G$ is said to be \emph{$(\ell,\varepsilon,\delta)$-compact} if, for any set $B\subseteq V$ of at least $\ell$ vertices, there is a subset $C \subseteq B$ of at least $\varepsilon\ell$ vertices that is a $\delta$-survival subset for~$B$.
\end{description}

Next we discuss properties of overlay graphs that make them suitable expanders.

\begin{description}
\item[\rm\em Expansion:]
graph~$G$ is \emph{$\ell$-expanding}, or is an \emph{$\ell$-expander}, if any two disjoint subsets of~$\ell$ vertices each are connected by an edge.
\end{description}

For a subset of vertices $W\subset V$, the \emph{edge boundary of~$W$}, denoted~$\partial W$, consists of the set edges in~$E$ that connect $W$  to $V\setminus W$. 
The minimum of the ratios $\frac{|\partial W|}{|W|}$, over all nonempty sets of vertices~$W$ such that $|W|\le \frac{|V|}{2}$, is this graph's \emph{edge expansion ratio}, denoted $h(G)$.
Regular graphs with sufficiently large constant vertex degree~$d$ and expansion ratio $h(G)>\frac{d}{3}$ can be constructed in polynomial time, see~\cite{DavidoffSV-book03,HooryLW06,KrebsS2011-book}.

\Paragraph{Structuring communication in overlay graphs.}

Nodes may communicate with their neighbors in suitable overlay graphs.
A major part of such communication is abstracted as local probing, which is an operation proposed and used in~\cite{ChlebusKS-PODC09} for similar purposes.
For local probing to work as intended, the overlay graph needs to have suitable expansion and compactness properties.
The purpose of local probing is to broadcast a value while testing the size of each node's generalized neighborhood of a suitably small radius for occurrence of faults.
We assume crashes in the context of local probing, unless stated otherwise.
There are two fixed integer parameters $\gamma$ and~$\delta$ that determine the duration and communication of local probing.
An instance of local probing takes exactly~$\gamma$ consecutive rounds.
Normally, each node $p$ sends a message to each neighbor at each among these~$\gamma$ rounds.
At the same time, normally a node $p$ receives messages from all its neighbors that have not crashed.
Node crashes may disrupt communication of non-faulty nodes as follows: if a node~$p$ receives fewer than $\delta$ messages at a round of local probing, then \emph{node~$p$ pauses prematurely}, in that it stops sending messages until $\gamma$ dedicated  rounds are over.
If a node $p$ has not paused prematurely during an execution of local probing, then \emph{node~$p$~survives} this instance of local probing.
Suppose that local probing  is applied with $B_1$ as its start set of operational   vertices  and  $B_2\subseteq B_1$ is its end set  of vertices that are operational just after the termination of local probing, while the vertices in~$B_1\setminus B_2$ crash during this local probing.
We summarize the properties of local probing that make it useful next a Proposition~\ref{pro:local-probing}, with a brief argument, given for completeness sake.

\begin{proposition}[\cite{ChlebusKS-PODC09}]
\label{pro:local-probing}

If there is a $(\gamma,\delta)$-dense-neighborhood for $p\in B_2$ in graph~$G|_{B_2}$,  then node~$p$ survives local probing. 
If there is no $(\gamma,\delta)$-dense-neighborhood for~$p\in B_1$ in graph~$G|_{B_1}$, then node~$p$ does not survive local probing.
Every node in a $\delta$-survival set~$C$ for~$B_2$ survives local probing.
\end{proposition}

\begin{proof}
Let $S\subseteq B_2$ be a $(\gamma,\delta)$-dense-neighborhood for $p$ in graph~$G|_{B_2}$. 
All $p$'s neighbors in~$S$ are operational, by the specification of~$B_2$. 
A node in~$S\cap N^{\gamma-1}_G(p)$ receives at least $\delta$ messages at the first round of local probing, as at least these many messages arrive from the neighbors.
This continues by induction on $i\le \gamma$, since no node in~$S\cap N^{\gamma-i}_G(v)$ pauses prematurely at the $i$-th round of local probing, so $p$ survives.
Suppose that $p$ survives local probing even though there is no $(\gamma,\delta)$-dense-neighborhood for~$p\in B_1$ in graph~$G|_{B_1}$.
There exists a set $S_1 \subseteq N_G(p)$ of at least $\delta$ vertices such that all the vertices in~$S_1$ survive the first $\gamma-1$ rounds of local probing.
For each $1\le i\le \gamma$, there is a set $S_i$ such that $S_{i-1} \subseteq S_i \subseteq N^{i}_G(p)$, and all the vertices in~$S_i$ survive the first $\gamma-i$ rounds of local probing, by induction on~$i$.
The set $S_{\gamma}$ satisfies the definition of $(\gamma,\delta)$-dense-neighborhood for $p$ in graph~$G|_{B_1}$.
This is a contradiction with lack of $(\gamma,\delta)$-dense-neighborhood for~$p\in B_1$ in graph~$G|_{B_1}$. 
Consider a $\delta$-survival set~$C$ for~$B_2$. 
Each node in~$C$ has at least~$\delta$ neighbors in~$C$, by the definition of  a $\delta$-survival set.
It follows by induction on the round number, that no node in~$C$ pauses prematurely at any round of local probing, since $C\subseteq B_2$. 
\end{proof}

Ramanujan graphs, discussed in Section~\ref{sec:ramanujan-graphs}, are  suitably compact and expanding.
These two properties suffice to derive the needed functionality of local probing if Ramanujan graphs determine the topology of overlay networks.

\section{Ramanujan Graphs}

\label{sec:ramanujan-graphs}

Graphs that are edge expanders can be used as overlay networks built into fault-tolerant distributed algorithms.
Such graphs are regular and parametrized by vertex degrees and the spectral gap.
The spectral gap defining Ramanujan graphs depends only on vertex degrees, which allows to discuss properties of Ramanujan graphs by referring to vertex degrees only.
We want vertex degrees in such regular graphs to be constant, in that that they do not depend on the number of vertices~$n$ and a bound on the number of faults~$t$. 

For a constant $d$, let $G=G(n, d)$ denote a $d$-regular Ramanujan  graph of $n$ vertices.
Let $\lambda_{1} \ge \lambda_{2} \ge \ldots \ge \lambda_{n}$ be the eigenvalues of $G(n,d)$. 
We use the notation $\lambda = \max(|\lambda_{2}|, |\lambda_{n}|)$. 
For $G=G(n,d)$ to be \emph{Ramanujan} means $\lambda \le 2\sqrt{d - 1}$.
We also use the following notations: 
\[
\ell(n,d) = 4nd^{-1/8}, \text{\ \ and \ \ } \delta(d) = \frac{1}{2}\bigl(d^{7/8} - d^{5/8}\bigr)
\ .
\]
Next, we show some properties of Ramanujan graphs that  support fault tolerance of networks with topologies of such graphs. 
In some of the proofs, we use the Expander Mixing Lemma, which is given as Lemma~2.5 in~\cite{HooryLW06}.

\begin{theorem}
\label{thm:ramanujan-expanding}

A Ramanujan graph $G(n,d)$ is $\ell(n,d)$-expanding.
\end{theorem}

\begin{proof}
Let $G(n,d)$ be a $d$-regular Ramanujan  graph of $n$ vertices such that $\lambda \le 2\sqrt{d - 1}$.
Take any two disjoint sets $A$ and $ B$ of vertices  such that $|A|, |B| \ge \ell(n,d) = 4nd^{-1/8}$.
Clearly $|A||B| \le n^2$. 
The Expander Mixing Lemma, applied to the two sets $A$ and $B$, yields
\[
\Bigl| e(A, B) - \frac{d|A||B|}{n} \Bigr| \le \lambda \sqrt{|A|\cdot|B|}
\ .
\]
This implies 
\[
e(A, B) \ge \frac{d|A| |B|}{n} - \lambda \sqrt{|A||B|}
\ .
\]
These estimates combined together give 
\[
e(A,B)\ge \frac{d}{n}\bigl(4\cdot nd^{-1/8}\bigr)^{2} - 2\sqrt{d-1}n \ge n\bigl(16\cdot d^{3/4} - 2\sqrt{d}\bigr)\ge
2 n \sqrt{d}(8 d^{1/4}-1)
\ , 
\]
which  is greater than~$0$ for  $d \ge 1$.
\end{proof}

Next, we show two auxiliary properties of Ramanujan graphs related to compactness. 
For a set of vertices~$S$ in $G(n,d)$, the number of edges in the subgraph of $G=G(n,d)$ induced by~$S$ is denoted by~$\text{vol}(S)$.

\begin{lemma}
\label{lemma:internal-lower-bound}

If a set of vertices $S$ in a Ramanujan graph $G=G(n,d)$ has at least $ nd^{-1/8}$ elements, then  $\mathrm{vol}(S) \ge  \delta(d) |S|$.
\end{lemma}

\begin{proof} 
Let $V$ be the set of vertices of $G$.
Suppose $S\subseteq V$ has  at least $nd^{-1 / 8}$ vertices. 
By the Expander Mixing Lemma applied to~$S$, we have the following inequality:
\[
\Bigl| e(S, V \setminus S) - \frac{d|S||V \setminus S|}{n} \Bigr| \le \lambda \sqrt{|S|\cdot|V \setminus S|}
\ . 
\]
It follows that
\[
e(S, V \setminus S) \le \lambda \sqrt{|S|\cdot|V \setminus S|} + \frac{d|S|\cdot|V \setminus S|}{n}
\ .
\]
A direct counting gives
\begin{equation}
\label{eqn:first-volume}
\text{vol}(S) \ge \frac{1}{2}\big(d \cdot |S| - e(S, V \setminus S) \big)
\ge \frac{1}{2}\bigg(d \cdot |S| - \lambda \sqrt{|S|\cdot|V \setminus S|} - \frac{d|S|\cdot|V \setminus S|}{n} \bigg).
\end{equation}
Observe that $|S| \cdot |V \setminus S| \le \frac{n^2}{4}$.
Since $G$ is Ramanujan, we have $\lambda \le 2\sqrt{d - 1} \le 2\sqrt{d}$. 
We can reformulate bound~\eqref{eqn:first-volume} as follows:
\begin{equation}
\label{eqn:second-volume}
\text{vol}(S) \ge \frac{1}{2}\big( d |S| - \sqrt{d} \cdot n - \frac{d|S|(n-|S|)}{n} \big) = \frac{1}{2}d \bigg(1 - \frac{n}{|S|}\frac{1}{\sqrt{d}} - \frac{n - |S|}{n} \bigg)  |S|.
\end{equation}
We have the inequality $\frac{n}{|S|} \le d^{1/8}$, since $|S| \ge nd^{-1/8}$, and also $\frac{n - |S|}{n} \le 1 - d^{-1/8}$. 
This leads to the following reformulation of inequality~\eqref{eqn:second-volume}:
\[
\text{vol}(S) 
\ge 
\frac{1}{2}d(1 - d^{-3/8} - \bigl(1 - d^{-1/8}\bigr) )|S|
= 
\frac{1}{2}d(d^{-1/8} - d^{-3/8})|S|
=
\frac{1}{2}  \bigl(d^{7/8} - d^{5/8}\bigr) |S|
\ ,
\]
which is the estimate sought.
\end{proof}

\begin{lemma}
\label{lem:bad-vertices}

If $X$ is a set of vertices in a Ramanujan graph $G=G(n,d)$, then there are fewer than $nd^{-1/8}$ vertices in $X$ each with the property that it has fewer than $\delta(d)$ neighbors in~$G|_{X}$.
\end{lemma}

\begin{proof}
Let $X'\subseteq X$ denote the set of vertices~$v$ such that $v$ has fewer  than~$\delta(d)$ neighbors in~$X$.
Suppose, to arrive at a contradiction, that $|X'|\ge nd^{-1/8}$.
By Lemma~\ref{lemma:internal-lower-bound}, the subgraph of~$G$ induced by~$X'$ has at least $|X'|\delta(d)$ edges. 
By the handshaking lemma and averaging, there is a vertex in~$X'$ with at least $2|X'|\delta(d)\cdot\frac{1}{|X'|}=2\delta(d)$ neighbors in~$X'\subseteq X$.
This is a contradiction, so we may conclude that $|X'| < nd^{-1/8}$.
\end{proof}

We show that compactness of Ramanujan graphs follows from suitable edge-density and edge-expansion properties,  by identifying a a $\delta(d)$-survival subset for any set of $\ell(n,d) = 4nd^{-1/4}$ vertices.

\begin{theorem}
\label{thm:ramanujan-compact}

A Ramanujan graph $G=G(n,d)$ is $(\ell(n,d), \frac{3}{4}, \delta(d))$-compact.
\end{theorem}

\begin{proof} 
Let $B$ be a set of at least $\ell(n,d)=4nd^{-1/8}$ vertices. 
We need to show that there is a subset $C \subseteq B$ of at least $\frac{3}{4}\ell(n,d)=3nd^{-1/8}$ vertices that is a $\delta(d)$-survival subset for~$B$, which means  that the degree of every vertex in~$G|_C$ is at least~$\delta(d)=\frac{1}{2}\bigl(d^{7/8} - d^{5/8}\bigr)$.

Let us define an operator~$F$ that acts on sets of vertices~of~graph~$G$ as follows:
\[
F_B(Y)=Y \cup \{v\in B\setminus Y \ | \ v \mbox{  has fewer than } 
\delta(d) \mbox{ neighbors in } B\setminus Y \}\ .
\]
This operator is monotonic with respect to inclusion, in the sense that if $Y_1\subseteq Y_2$ then also the inclusion $F_B(Y_1)\subseteq F_B(Y_2)$ holds. 

Consider a sequence $\langle Y_i\rangle_{i\ge 0}$ defined by the recurrence $Y_{i+1}=F_B(Y_i)$, with $Y_{0}=\emptyset$.
In particular, $Y_1\subseteq B$ is a set of these vertices that have fewer  than~$\delta(d)$ neighbors in~$B$. 
Since the sequence $\langle Y_i\rangle_{i\ge 0}$ is monotonic, discrete and each term is included in $B$, the set $B^*= \bigcup_{i\ge 0} Y_i$ is a fixed point, in the sense that 
$B^*=Y_k$ for all sufficiently large natural numbers~$k$. 

The set $C=B\setminus B^*$ is a natural candidate for a $\delta(d)$-survival set of~$B$. 
Each degree of a vertex in~$G|_C$ is at least~$\delta(d)$, because otherwise this vertex would have been added to $Y_{i+1}=F_B(Y_i)$ in some $i$th iteration of the operator~$F$.
To complete demonstrating the claim of $(\ell(n,d), \frac{3}{4}, \delta(d))$-compactness, we need to show that $B^*=B\setminus C$ has at most $nd^{-1/8}$ elements.

Suppose, to the contrary, that $|B^*|> nd^{-1/8}$.
Let $Y_j$ be the first set in the sequence $\langle Y_i\rangle_{i\ge 0}$ that has at least $nd^{-1/8}$ elements. 
Such a set exists since $|B^*| >nd^{-1/8}$.
We consider the number of edges  that have both endpoints in the set $Y_{j}$, whose number is denoted as $\text{vol}(Y_{j})$. 
By Lemma~\ref{lemma:internal-lower-bound}, $\text{vol}(Y_{j}) \ge \delta(d) |Y_{j}|$.
To arrive at a contradiction, we show next that $\text{vol}(Y_j) < \delta(d) |Y_j|$. 
If $j=1$ then this follows from Lemma~\ref{lem:bad-vertices} by setting $X=B$.
Suppose that $j>1$, so that the iterations of operator $F_B(Y)$ generate the sequence $Y_1, Y_2,\ldots, Y_j$.
Let an edge~$e$ be in the subgraph induced by~$Y_j$. 
If $e$ has an endpoint in~$Y_1$, then it has been accounted for in $\delta(d) |Y_1|$.
Otherwise, let $k$ be such that $1<k\le j$ and $k$ is the smallest among indices~$i$ of $Y_i$ with the property that an endpoint~$v$  of this edge~$e$ is in~$Y_i$. 
Then this edge is accounted for in the number of neighbors of~$v$ outside $Y_{i-1}$, which is less than~$\delta(d)$, by the specification of the operator~$F$.
As we place vertices in~$Y_i$ from outside of~$Y_{i-1}$, each added vertex contributes a number of neighbors that is less than~$\delta(d)$.  
So the sum that gives $\text{vol}(Y_j)$ is less than $\delta(d) |Y_j|$.
This provides a contradiction and completes a proof of the inequality $|B^*|\le nd^{-1/8}$.
\end{proof}

\begin{theorem}
\label{thm:sparse-dense}	

An $(\gamma(n),\delta(d))$-dense-neighborhood of a vertex in a Ramanujan graph~$G(n,d)$ includes at least $\ell(n,d)$ vertices, if only $\gamma(n) \ge 2+\lg n$ and $d$ is sufficiently large.
\end{theorem}

\begin{proof}
Let $G(n,d)=(V,E)$.
Consider any vertex $v \in V$ and its $(\gamma(n),\delta(d))$-dense-neighborhood $S \subseteq N^{\gamma(n)}_{G}(v)$. 
We show by induction that the set $A_{i} = S \cap N^{i}_{G}(v)$ has at least $\min (2^{i}, \ell(n,d) )$ vertices, for all $i$ such that $1 \le i \le \gamma(n) - 1$.
This suffices, since $2^{2+\lg n}=4n\ge \ell(n)$.

For the base case $i = 1$, observe that $N^{1}_{G}(v)$ is the neighborhood of a vertex~$v$ in~$S$, which has at least $2$ neighbors, since $G(n,d)$ is $d$-regular for $d>1$.
Suppose that $i \ge 2$, and  the inductive hypothesis holds for~$i$.
Let $Z = (N_{G}(A_{i}) \cap S)\setminus A_i$ be the part of the neighborhood of $A_i$ in the subgraph induced by~$S$ that extends beyond~$A_i$. 
Observe that $A_{i + 1} = A_{i} \cup Z$ and $|Z\cup A_i|=|Z|+|A_i|$.

If $|A_{i}| + |Z| > \ell(n,d)$, then $A_{i + 1}$ has more than $\ell(n,d)$ vertices. 
Suppose otherwise that the inequality $|A_{i}| + |Z| \le \ell(n,d)$ holds. 
For each set $S$ of vertices of graph $G(n,d)$ of  at most $\beta n$ elements, the number of edges in the subgraph induced by $S$ is at most $\frac{1}{2}d\beta^{2}n + \sqrt{d - 1}\beta(1 - \beta)n$, as shown by Alon and Chung~\cite{AlonC06}.
Applying this with $\beta = \frac{|A_{i} \cup Z|}{n}$, we obtain that
\begin{equation}
\label{eqn:dense}
\text{vol}(A_{i} \cup Z) \le \frac{1}{2}d\beta^{2}n + \sqrt{d-1}\beta(1-\beta)n \le \frac{1}{2}d\beta|A_{i} \cup Z| + \sqrt{d}|A_{i} \cup Z|\ .
\end{equation}
If $|A_{i} \cup X| \le 4nd^{-1/4}n$ then $\beta\le 4nd^{-1/4}$, and
we can upper bound the right-hand side of the  inequality~\eqref{eqn:dense} as follows:
\begin{equation}
\label{eqn:vol}
\frac{1}{2}d\beta|A_{i} \cup X| + \sqrt{d}|A_{i} \cup X| \le |A_{i} \cup X|\big(2d^{3/4} + \sqrt{d}\big)\ .
\end{equation}
Every vertex in $A_{i}$ has degree at least  $\delta(d)$ in the graph induced by vertices from $A_{i} \cup Z$, so  $\text{vol}(A_{i} \cup Z) \ge \frac{\delta(d)}{2}|A_{i}|$ by the handshaking lemma.
This combined with the inequality~\eqref{eqn:vol} gives 
\[
\frac{\delta(d)}{2}|A_{i}| \le \text{vol}(A_{i} \cup Z) \le |A_{i} \cup Z|\big(2d^{3/4} + \sqrt{d}\big)\ , 
\]
and hence also 
\[
\frac{\delta(d)}{2}|A_{i}| 
\le 
|A_{i} \cup Z|\big(2d^{3/4} + \sqrt{d}\big) 
\le
\big(|A_{i}| + |Z| \big)\big(2d^{3/4} + \sqrt{d}\big)\ .
\]
Solving for $|Z|$ gives
\[
|Z|\ge \Bigl(\frac{\delta(d)}{4d^{3/4} + 2\sqrt{d}} - 1 \Bigr)|A_{i}| 
\ .
\] 
Since $\delta(d) = \frac{1}{2}(d^{7/8} - 2d^{5/8})$,
the following inequality holds for a sufficiently large~$d$:
\[
\frac{\delta(d)}{4d^{3/4} +2\sqrt{d}} - 1 \ge 1
\ .
\]
We conclude with the estimate 
\[
|A_{i+1}| = |Z| + |A_{i}| \ge 2|A_{i}| \ge 2^{i + 1}\ ,
\]  
which completes the inductive step.
\end{proof}

\begin{theorem}
\label{thm:ramanujan-unbalanced-expansion}

Let $0< \epsilon <1$ be a fixed constant, and let $A$ and $B$ be two disjoint subsets of vertices of a Ramanujan graph $G=G(n,d)$. 
If $|A| =  \epsilon \cdot n$ and $|B| > \frac{4n}{d\epsilon}$, then there exists an edge connecting the sets $A$ and $B$ with each other. 
\end{theorem}

\begin{proof} 
Consider any such disjoint sets of vertices $A$ and $B$  in~$G(n,d)$. 
The Expander Mixing Lemma applied to sets $A$ and $B$ gives the following estimate:
\[
\Bigl| e(A, B) - \frac{d|A||B|}{n} \Bigr| \le \lambda \sqrt{|A|\cdot|B|}
\ .
\]
This implies $e(A, B) \ge \frac{d|A|\cdot|B|}{n} - \lambda \sqrt{|A|\cdot|B|}$.
Since $\lambda \le 2\sqrt{d - 1} \le 2\sqrt{d}$ and $|A| = \epsilon n$, we obtain
\begin{equation}
\label{ineq:2}
e(A,B) \ge \frac{d|A|\cdot|B|}{n} - \lambda \sqrt{|A|\cdot|B|} \ge d\epsilon|B| - 2\sqrt{d}\sqrt{|A||B|} = d|B|\Bigl(\epsilon - \frac{2}{\sqrt{d}}\sqrt{\frac{|A|}{|B|}} \Bigr)
\ .
\end{equation}
The assumption $|B| > \frac{4n}{d\epsilon}$ yields 
$|B| > \frac{4n}{d\epsilon} = \frac{4|A|}{d\epsilon^{2}}$.
This implies $\frac{d\epsilon^{2}}{4} > \frac{|A|}{|B|}$, and solving for $\epsilon$ gives $\epsilon > \frac{2}{\sqrt{d}}\sqrt{\frac{|A|}{|B|}}$.
Combining  this estimate of~$\epsilon$ with the lower bound in~\eqref{ineq:2}, we obtain $e(A,B)  > 0$, so that an edge between $A$ and $B$ exists.
\end{proof}



\section{Consensus with Crashes}

\label{sec:agreement-crashes}

We develop algorithms  for the problems of almost everywhere agreement, spreading a common value, and consensus.
Crashes are the model of failures.
The input values are either $0$ or~$1$.

Messages used by the algorithms have two possible roles: either to carry a (rumor) value or to inquire for a decision value. 
The role of a message is determined by the round in which it is sent, so it suffices for a message to carry one bit of information.

\subsection{Almost Everywhere Agreement}

An algorithm we give for almost everywhere agreement  is called \textsc{Almost-Everywhere-Agreement} and its pseudocode is in Figure~\ref{fig:almost-everywhere-agreement}.
The algorithm uses a graph denoted $G$ which serves as an overlay network.
We assume that $t<\frac{n}{5}$ in this sub-section.
Nodes with the smallest $5t$ names are called \emph{little}.
There are exactly $5t$ little nodes, by the assumption $t<\frac{n}{5}$.
The graph $G$  has the little nodes as vertices and its topology is of a $G(5t, d)$ Ramanujan graph with the vertex  degree~$d = 5^8$. 
This degree~$d$ determines the quantities $\ell =  4\cdot 5t d^{-1/8}=4t$ and $\delta = \frac{1}{2}(d^{7/8} - 2d^{5/8})$ that are relevant to the local probing in Part~2 in the pseudocode in Figure~\ref{fig:almost-everywhere-agreement}.

\begin{figure}[t!]

\hrule

\FF

\texttt{Algorithm} \textsc{Almost-Everywhere-Agreement}

\FF

\hrule

\

\begin{description}[nosep,leftmargin=1em]
\item[\sf Part~1: Broadcasting initial value $1$:] \ 

\texttt{if} $p$ is little  \texttt{then} 
\begin{enumerate}[nosep,leftmargin=1em]
\item[]
set the candidate decision value to the input value
\item[]
\texttt{for} $5t-1$ rounds \texttt{do}
\begin{enumerate}[nosep,leftmargin=1em]
\item[]
\texttt{if} either $p$'s candidate decision is $1$ and it is the first round \texttt{or} \\
$p$'s candidate decision is $0$ and $p$ received rumor $1$ in the previous  round for the first time 
\texttt{then} 
set the candidate decision value to $1$
and
send rumor $1$ to each neighbor in~$G$
\end{enumerate}
\end{enumerate}
\texttt{else} stay idle for $5t-1$ rounds

\item[\sf Part~2: Local probing:] \ 

\texttt{if} $p$ is little \texttt{then}  run local probing on graph $G$ for $2+\lg (5t)$ rounds, such that
\begin{enumerate}[nosep,leftmargin=1em]
\item[]
(a) send rumor in messages
\item[]
(b) \texttt{if} the candidate decision value is~$0$ and rumor $1$ received \texttt{then} set candidate decision to~$1$
\end{enumerate}
\texttt{if} $p$ is little and survived local probing \texttt{then} decide on the candidate decision value

\item[\sf Part~3: Notifying related nodes of decision:] \ 

\texttt{if} $p$ is little and has decided \texttt{then} notify each related node of the decision

\texttt{if} $p$ received a decision value from a little neighbor  \texttt{then} decide on the received value

\end{description}

\FF

\hrule

\caption{\label{fig:almost-everywhere-agreement}
A pseudocode for a node~$p$.
Node $p$ is little if its name is at most $5t$.
The edges of a regular graph~$G$ determine links to send messages in the first two parts.
The quantities $\ell$ and~$\delta$ relevant to local probing are determined by the vertex degree~$d$ of $G$.
}

\end{figure}

Each among the nodes maintains its candidate decision value throughout an execution,  which is referred to as  \emph{rumor} when sent in a message.
At the start, each little node sets its candidate decision value to its initial input value.
The algorithm proceeds through three parts.

At the first round of Part~1, each little node with rumor~$1$ transmits a message with the rumor to all neighbors in graph~$G$.
In the following rounds of Part~1, if a little node~$p$ receives a message with rumor~$1$, while the current candidate decision value is~$0$, then $p$ sets its candidate decision value to~$1$ and notifies all its neighbors in~$G$ of the change by sending rumor~$1$ to them.
This process of flooding the sub-network of little nodes with rumor value~$1$ continues for~$5t-1$ rounds.

Part~2 follows next.
It consists of local probing in graph~$G$ of little nodes for $2+2\lg (5t)$ rounds.
The little nodes send their candidate decision value settled during Part~1.
To make the algorithm precise, we stipulate that if a node has $0$ as a candidate value and receives rumor~$1$ then it changes the candidate value to~$1$ and uses it as rumor from that round on.
Little nodes that survive local probing decide on their candidate values.

Starting from Part~3, all nodes participate.
If a node's name~$j$ is congruent modulo $5t$ to a little node of name~$i$ then we say that $i$ and $j$ are \emph{related}.
Part~3 consists of all little nodes that have decided notifying their respective related nodes of the decision, and the recipients of such messages adopt the received decision.

\begin{lemma}
\label{lem:few-end-part2-number}

Assuming $t<\frac{n}{5}$, at least $3t$ non-faulty little nodes decide while executing Part~2 of algorithm \textsc{Almost-Everywhere-Agreement}.
\end{lemma}

\begin{proof}
There are at least $5t$ little nodes, since $t<\frac{n}{5}$.
Let a set $B$ consist of the non-faulty little nodes in an execution.
There are at least $5t-t=4t$ elements in~$B$.
Graph $G$ is a $G(5t, d)$ Ramanujan graph for $d = 5^{8}$. 
We defined $\ell =  20 t d^{-1/8}$, which equals  $4t$ by the determination of degree~$d$.
Graph $G$ is $(\ell, \frac{3}{4}, \delta)$-compact, for $\delta = \frac{1}{2}(d^{7/8} - 2d^{5/8})$, by Theorem~\ref{thm:ramanujan-compact}.
It follows, by the definition of compactness, that there exists a $\delta$-survival subset of~$B$ of at least $\frac{3}{4}\ell$ elements. 
%
By Proposition~\ref{pro:local-probing}, every little node in such a $\delta$-survival set survives local probing in Part~2. 
Little nodes that survive local probing decide, by the pseudocode in Figure~\ref{fig:almost-everywhere-agreement}.
It follows that at least $\frac{3}{4}\ell=3t$ non-faulty little nodes decide.
\end{proof}

The pseudocode in Figure~\ref{fig:almost-everywhere-agreement} stipulates about Part~2 that if the candidate decision value is~$0$ and rumor $1$ is received then then the candidate decision value is reset to~$1$.
Next, we show that if $t<\frac{n}{5}$ then such resetting never occurs.

\begin{lemma}
\label{lem:few-end-part2-validity}

Assuming $t<\frac{n}{5}$, if a node survives local probing in Part~2 of \textsc{Almost-Everywhere-Agreement} then it did not change its candidate decision from $0$ to~$1$ during local probing.
\end{lemma}

\begin{proof}
The non-faulty little nodes induce a subgraph of $G$.
The non-faulty little nodes in the same connected component of the induced subgraph set their rumors to the same value by the end of Part~$1$.
This is because the length of a longest path is $5t-1$, so $5t-1$ round suffices to propagate through the connected component.
%
%
A node that survives local probing belongs to a $(2+\lg (5t),\delta(d))$-dense-neighborhood, by the duration of local probing in Part~2, and the mechanism of local probing determined by~$\delta$.
By Theorem~\ref{thm:sparse-dense}, $(\gamma(5t),\delta(d))$-dense-neighborhood of a vertex includes at least $\ell(5t,d)$ vertices, if $\gamma(5t) \ge 2+\lg (5t)$.
%
%
Any two such neighborhoods are connected by an edge, by Theorem~\ref{thm:ramanujan-expanding}.
It follows that nodes that belong to such neighborhoods acquire equal candidate values by the end of Part~1, and so use the same rumor in local probing in Part~2.
\end{proof}

\begin{theorem}
\label{thm:AEA}

Assuming $t<\frac{n}{5}$, algorithm \textsc{Almost-Everywhere-Agreement} solves the problem $\frac{3}{5}$-AEA in~$\cO(t)$ rounds and with $\cO(n) $  messages, each message carrying one bit.
\end{theorem}

\begin{proof}
At least $3t$ non-faulty little nodes decide in Part~2, by Lemma~\ref{lem:few-end-part2-number}.
Each such a node notifies its $\frac{n}{5t}$ related nodes in Part~3.
A node that is not little is related to precisely one little node.
Therefore at least $\frac{3n}{5}$ nodes either decide or crash by the end of Part~3.

Validity follows from the fact that rumors circulated among the nodes originate as initial values of nodes.
Part~3 serves the purpose to spread decisions made at the end of Part~2. 
To show agreement, it suffices to demonstrate that all  decisions made by little nodes by the end of Part~2 are equal.
Such decisions are taken after surviving local probing in Part~2.
By Lemma~\ref{lem:few-end-part2-validity}, nodes that decide do not change their candidate decision values during Part~2, so the candidate decision values settled during Part~1 are final.
The rumors spread during Part~1 carry only value~$1$, so either all the nodes that decide in Part~2 set their candidate decision value to~$1$ during Part~1 or none does it.
This gives agreement, which means that no two nodes decide on different values.
Termination follows from a bound on the duration of the execution determined by the pseudocode in Figure~\ref{fig:almost-everywhere-agreement}.

This bound on running time  is as follows.
Part~$1$ takes $5t-1$ rounds, Part~$2$ takes $2+\lg (5t)$ rounds, and Part~3 takes one round, for a total number of rounds that is $\cO(t)$.
The number of sent messages is bound as follows.
The graph $G$ is of degree $d=5^8$, so in Part~1, each little node sends at most $5^8=O(1)$ messages.
In Part~2, each of the $5t$ little nodes sends at most $5^8(2+\lg (5t))=\cO(\log t)$ messages.
 In Part~3, each among $5t$ little nodes sends at most $\frac{n}{5t}$ messages, for a total of $n$ messages. 
 The total number of messages is $\cO(n)$.
 Each of these messages carries a single value of either $0$ or $1$.
\end{proof}

\begin{figure}[t!]

\hrule

\FF

\texttt{Algorithm} \textsc{Spread-Common-Value} 

\FF

\hrule

\FF

\begin{description}[nosep,leftmargin=1em]

\item[\sf Part~1: Broadcasting common value:] \ 

\texttt{if} $p$ decided \texttt{then} send common value to all $p$'s  neighbors in graph~$H$

\texttt{for} $\lceil \log_{\frac{3}{2}} \frac{2n/5}{\max\{t,n/t\}} \rceil $ rounds \texttt{do}
\begin{enumerate}[nosep,leftmargin=1em]
\item[]
\texttt{if} $p$ has not decided yet \texttt{and} received a common value in the previous round \texttt{then}

\begin{enumerate}[nosep,leftmargin=1em]
\item[]
decide on the received common value \texttt{and} forward  it to all $p$'s neighbors in $H$ 
\end{enumerate}
\end{enumerate}

\item[\sf Part~2: Inquiring about common value:] \

\texttt{if} $t^2\le n$ \texttt{then}

\begin{enumerate}[nosep,leftmargin=1em]
\item[]
\texttt{if} $p$ has not decided yet \texttt{then} send an inquiry to every little node

\texttt{if} $p$ has decided \texttt{and} 
an inquiry  received  \texttt{then} 
\begin{enumerate}[nosep,leftmargin=1em]
\item[]
respond to each inquiring node with~a common value
\end{enumerate}
\texttt{if} a response for $p$'s own inquiry arrived \texttt{then} adopt the received value as a common value 
\end{enumerate}

\texttt{else}  \texttt{for} $i=1$ \texttt{to} $\lceil \lg (t+1) \rceil$ \texttt{do} perform Phase $i$: 
\begin{enumerate}[nosep,leftmargin=1em]
\item[]
at round $2i-1$ of this part: 
\begin{enumerate}[nosep,leftmargin=1em]
\item[]
\texttt{if} not decided yet \texttt{then} send an inquiry to each $p$'s neighbor in graph $G_{i}$
\end{enumerate}
\item[]
at round $2i$ of this part: \

\begin{enumerate}[nosep,leftmargin=1em]
\item[]
\texttt{if} decided \texttt{and} an inquiry from a neighbor in graph $G_{i}$ received  at round $2i-1$
\begin{enumerate}[nosep,leftmargin=1em]
\item[]
\texttt{then} respond to each inquiring neighbor with the common value
\end{enumerate}
\item[]
\texttt{if} a response for own inquiry arrived \texttt{then} adopt the received value as common value 

\end{enumerate}
\end{enumerate}

\end{description}

\FF

\hrule

\FF

\caption{\label{fig:spread-value}
A pseudocode for a node~$p$.
The topologies of the overlay graphs denoted as $H$ and~$G_i$ determine links to send messages.
A node decides on the common value as soon as it learns it.
}
\end{figure}

\subsection{Spreading a Common Value}

We give an algorithm for spreading a common value.
The algorithm is called \textsc{Spread-Common-Value}, and its pseudocode is in Figure~\ref{fig:spread-value}.
We assume that $t<\frac{n}{5}$ in this sub-section.
Additionally, the constants  in the pseudocode and overlay graphs are selected to work if at least $\frac{3}{5}n$ nodes have been initialized with a common value.
Each node uses the same dedicated variable to be initialized either  to the common value or  to \texttt{null}.
A node decides on the common value by setting the dedicated variable to the common value. 
A graph denoted by~$H$ in the pseudocode in Figure~\ref{fig:spread-value} is a Ramanujan graph of degree $\Delta \ge 64$ with an edge expansion $h(H)$ that is at least $\frac{1}{2}(\Delta - 2\sqrt{\Delta - 1}) \ge \frac{\Delta}{3}$.
For a set of vertices $A$ of a graph, the \emph{external neighbors of $A$} are defined as these vertices connected by an edge to a vertex in $A$ that do not belong to~$A$.
A graph denoted as $G_i$ in the pseudocode in Figure~\ref{fig:spread-value} 
has properties as stated in Lemma~\ref{lem:graph-G-i}.

\begin{lemma}
\label{lem:graph-G-i}

For a sufficiently large $n$ and assuming $t<\frac{n}{5}$, for any $i$ such that $1\le i\le \log (t+1)$, there exists a graph~$G_i$ of vertex degree at most $\cO( 2^{i+1})$, such that any set of $C\cdot \frac{t+1}{2^i}$ vertices has at least $2(t+1)$ external neighbors, for some constant $C>0$.
\end{lemma}

\begin{proof}

We denote $2^{i} \cdot 10$ as~$b_i$ and $\frac{t+1}{2^i}$ as~$c_i$.
Let $t<\frac{n}{5}$ and $1\le i\le \log (t+1)$.
We want to prove that a suitable random graph satisfies the needed properties of  graph $G_i$ with a positive probability.
Such a random graph is defined as follows: every vertex~$v$ chooses to have another vertex~$w$ as its neighbor depending on an outcome of a Bernoulli trial for the ordered pair~$(v,w)$ with probability~$\frac{b_i}{n}$ of success.
Two vertices $v$ and $w$ are connected by an edge if either vertex~$v$ chooses $w$ as its neighbor, or vertex~$w$ chooses $v$ as its neighbor, or both.

There are $\binom{n}{c_i}$ and $\binom{n}{t}$ different subsets of vertices of size $c_i$ and $t$, respectively. 
There are at most  $\binom{n}{c_i}\cdot \binom{n}{t}\le \binom{n}{t+1}^2$ pairs of such sets, since $c_i\le t+1 < \frac{n}{5}+1<\frac{n}{2}$.
Fix such a pair of $C_i$ and $A$. 
Let us consider the event $\cE(C_i,A)$ which holds when some vertex in $C_i$ does not have any neighbor outside of~$C_i\cup A$. 
The probability of the event $\cE(C_i,A)$ is at most 
\[
\left(\frac{c_i+t}{n}\right)^{b_i c_i}
\ge
\left(\frac{2t+1}{n}\right)^{10(t+1)}
\ .
\]
This quantity is upper bounded by the probability that during the process of each vertex in $C_i$ selecting its neighborhood and $c_i\le t+1$, all neighbors are in $C_i\cup A$.
The probability that a given set $C_i$ has at most $2t+1$  external neighbors
is upper bounded by the union bound, over all possible sets of vertices $A$ of size $t$, of events $\cE(C_i,A)$. This is at most
\[
\binom{n}{t} \cdot \left(\frac{2t+1}{n}\right)^{10(t+1)}
\ .
\]
The probability that there is a set $C_i$ of $c_i$ vertices with at most $2t+1$ external neighbors is at most 
\[
\binom{n}{c_i}\cdot \binom{n}{t} \cdot \left(\frac{2t+1}{n}\right)^{10(t+1)}
\le
\binom{n}{t+1}^2 \cdot \left(\frac{2t+1}{n}\right)^{10(t+1)}
\le
\left(\frac{n e}{t+1}\right)^{2(t+1)}
\cdot \left(\frac{2t+1}{n}\right)^{10(t+1)}
\]
\[
\le
\left(\frac{e2^4(t+1)^4}{n^4}\right)^{2(t+1)}
< 1
\ ,
\]
as $t+1<\frac{n}{2\sqrt[4]{e}}$.
The complementary event, which is the sought property of graph $G_i$, holds with a positive probability. 
By the probabilistic-method argument, there exists such a graph~$G_i$.
\end{proof}

Algorithm \textsc{Spread-Common-Value} consists of two parts, as represented in Figure~\ref{fig:spread-value}.
The common value is broadcast through graph~$H$ for $\lceil \log_{\frac{3}{2}} \frac{2n/5}{\max\{t,n/t\}} \rceil$  rounds in Part~1. 
A node that has already decided notifies its neighbors of the common value, which they adopt as their own decision.


\begin{lemma}
\label{lem:part-one-SCV}

If at least $\frac{3}{5}n$ nodes have been initialized with a common value then at most $\max\{Ct, t + \frac{n}{t}\}$ other nodes may not decide by the end of an execution of Part~1 of algorithm \textsc{Spread-Common-Value}, for some constant $C$.
\end{lemma}

\begin{proof}.
We refer to rounds of Part~1 of \textsc{Spread-Common-Value} numbered from~$1$. 
Let $w_k$ be the number of nodes that have not received a common value by round~$k$. 
Since at least $\frac{3}{5}$ nodes has been initialized with a common value, we have that $w_1\le \frac{2n}{5}$.
Let $\ell_k$ be the number of nodes that receive a common value at round~$k+1$.
The overlay graph~$H$ used in 
Part~1 of \textsc{Spread-Common-Value} is a Ramanujan graph of degree $\Delta \ge 64$ with edge expansion $h(H)$ that is at least $\frac{1}{2}(\Delta - 2\sqrt{\Delta - 1}) \ge \frac{\Delta}{3}$.
There are at least these many edges outgoing from the nodes that have not received a common value yet by round $k$ in graph~$H$, discounting for up to~$t$ crashes:
$h(H) \cdot (w_k -t) \ge \frac{\Delta}{3} \cdot (w_k-t)$.
The number of non-faulty neighbors connected by these outgoing edges is at least these many: $\ell_{k}\ge \frac{1}{\Delta}\cdot  \frac{\Delta}{3} \cdot (w_k - t) = \frac{1}{3}\cdot (w_k-t)$.
Each such a non-faulty neighbor has already received a common value so it sends a message with the value.
This gives a recurrence: $w_{k+1} \le w_k-\ell_k=\frac{2}{3}w_k +\frac{t}{3}$.
Iterating this recurrence, we obtain that for $j=\lceil \log_{\frac{3}{2}} \frac{2n/5}{\max\{t,n/t\}} \rceil$ the number of nodes that have received the value yet is at most $w_j\le \max\{t,\frac{n}{t}\} + \frac{t}{3}\sum_{k\ge 0} (\frac{2}{3})^i=\max\{t,\frac{n}{t}\}+t=\max\{2t,t+\frac{n}{t}\}$, which gives the lemma for $C = 2$.
\end{proof}
  
Part~2 of algorithm \textsc{Spread-Common-Value} is spent by the nodes that have not decided yet inquiring other  nodes  about their decisions.
If a node that has already decided receives such an inquiry then it responds by sending a message with the common value.
A successful inquiry by a node that has not decided yet results in a response that brings the common value, which is then adopted as own decision.
Part~2 is structured into $2+\lceil \lg t \rceil$ phases such that a node that has not decided yet by phase~$i$ sends inquiries to its neighbors in  the overlay graph~$G_i$.

\begin{lemma}
\label{lem:constant-time-part-five}

If $t^2\le n$ then all non-faulty nodes receive a common value in the if-case of Part~2 of \textsc{Spread-Common-Value}.
\end{lemma}

\begin{proof}
If a node has not received the value yet then it sends inquiries to all little nodes, of which at least $3t$ are non-faulty that have decided already and therefore have the common value, c.f. Lemma~\ref{lem:few-end-part2-number}.
Each of them sends a response with a value, which is then adopted as own.
\end{proof}

\begin{lemma}
\label{lem:part-two-SCV}

If $\sqrt{n} < t<\frac{n}{5}$ and at least $\frac{3}{5}n$ nodes have been initialized with a common value  then  fewer than $C \frac{t+1}{2^{i}}$ non-faulty nodes do not have a common value by Phase~$i$ of Part~2 of \textsc{Spread-Common-Value}, for each~$i$ such that $1\le i\le \lceil \lg (t+1) \rceil$ and a constant $C$.
\end{lemma}

\begin{proof} 
To simplify the inductive analysis, let us assume that a ``dummy'' Phase~$0$ ends at the very beginning  of Part~1 of \textsc{Spread-Common-Value}.
The claim holds for this phase by Lemma~\ref{lem:part-one-SCV}, which holds after Part~1 and before Part~2 starts.
Suppose, to the contrary, that the claim of current lemma does not hold for some $i$ such that $1\le i\le \lceil \lg t \rceil$. 
This means that there exists a set~$C_i$ of non-faulty nodes that have not received a common value yet by round~$2i$, such that $|C_i|=\frac{t}{2^{i+1}}$,
for some~$i\ge 1$. 
The overlay graph~$G_{i}$, for $i\ge 1$, used 
in Phase~$i$ by nodes in $C_i$ for sending inquiries and receiving responses is a graph of vertex expansion $\left(1-\epsilon\right)2^{i}\cdot 10$. We also have that $|C_{i}| = \frac{t}{2^{i + 1}} \le \frac{n}{\left(1 - \alpha\right)\left(1-\epsilon\right)2^{i}\cdot 10}$ since $t < n/5$. 
By the vertex expansion property applied to set $C_i$, we conclude that this set has at least $\left(1-\epsilon\right)2^{i}\cdot 10 |C_{i}| = \left(1-\epsilon\right) 5t$ neighbors in graph $G_{i}$. 
By Lemma~\ref{lem:part-one-SCV}, at most $2t$ nodes do not decide by the end of Part~1. Adding additional slack of $t$ for the number of faulty nodes, still allows us to reason that at least one neighbor of $C_{i}$ has to decide before the beginning of Part~2, and therefore, the decision was relayed by that 
process to some of its neighbor processes in~$C_i$ during Phase~$i$. 
It follows that some process in $C_i$ decides 
by the end of Phase~$i$, which is a contradiction with the definition of~$C_i$ and thus validates the claim of the lemma.

\end{proof}

\begin{theorem}
\label{thm:SCV}

Assuming $t<\frac{n}{5}$, algorithm \textsc{Spread-Common-Value} solves $\frac{3}{5}$-SCV in $\cO(\log t) $ rounds and sending $ \cO(t \log t)$ messages.
\end{theorem}

\begin{proof}
For termination, we show that every node decides by the end of Part~2, unless it crashes.
If $t^2\le n$ then all nodes decide by Part~2, by Lemma~\ref{lem:constant-time-part-five}.
If $t^2> n$ then $2+\lg t$ phases get executed.
By Lemma~\ref{lem:part-two-SCV}, fewer than $\frac{t}{2^{i-1}}$ non-faulty nodes do not decide by the end of phase~$i$, for $i$ such that $1\le i\le 2+\lceil \lg t \rceil$.
The number $\frac{t}{2^{2+\lceil \lg t \rceil -1}}$ is less than $1$, so the set of nodes that have not decided by phase number $2+\lceil \lg t \rceil$ of Part~$2$ is empty.

Part~1 takes $\lceil \log_{\frac{3}{2}} \frac{2n/5}{\max\{t,n/t\}} \rceil $ rounds, which is at most $2\lg t$, and
Part~2 takes at most $2+\lg t$
rounds.
All together we have at most $5t+ 4(1+\lg t)\le 5(t+\lg t)=\cO(t)$ rounds. 
Each node is active in at most $2$ rounds in Part 1, when receiving/sending for the first time only, the whole Part 2 of length $O(\log t)$ and Part 3 of one round.
This gives $\cO(t)$ active rounds in total.
 In Part~1, the graph $H$ has degree $64$, so each node sends at most $64=O(1)$ messages.
 There are at most $\max\{2t,t+\frac{n}{t}\}$ nodes that execute Part~2.
If $t^2\le n$ then each of these at most $t+\frac{n}{t}$ nodes sends messages to every among $5t$ little nodes, for a total of $5t^2 +n\le 6n=O(n)$ messages.
If $t^2> n$ then nodes use overlay graphs $G_i\in\cG$, in which nodes have degree 
$2^{i}\cdot 10$, for sufficiently large $n$.
By Lemma~\ref{lem:part-two-SCV}  (there is a new version), in this case nodes send at most these many messages in Part~2:
\[ 
\sum_{i=0}^{\ceil{\lg{t}}} \frac{t}{2^{i - 1}} \cdot 
2^{i}\cdot 10
\le 
20t\lg{t}
= O(t \log t)
\ .
\]
All these messages are of constant size.

\end{proof}

\remove{
\begin{proof}
The nodes communicate the common value in messages, so this is the only possible value to decide on.
For termination, we show that every node decides by the end of Part~2, unless it crashes.
By Lemma~\ref{lem:part-two-SCV}, fewer than $\frac{t+1}{2^{i}}$ non-faulty nodes do not decide by the end of phase~$i$, for $i$ such that $1\le i\le \lceil \lg (t+1) \rceil$.
The number $\frac{t+1}{2^{\lceil \lg (t+1) \rceil }}$ is at most~$1$, so the set of nodes that have not decided by the last phase $\lceil \lg (t+1) \rceil$ of Part~$2$ is empty.
Part~1 takes $\lceil \log_{\frac{3}{2}} \frac{2n}{5}\rceil +1$ rounds, 
and Part~2 takes $\lceil \lg (t+1) \rceil$ rounds, which together make $\cO(\log n)$ rounds.
In Part~1, the graph $H$ has a constant degree, so each node sends at most $\cO(1)$ messages.
By Lemma~\ref{lem:part-one-SCV}, at most $t+1$ nodes execute Part~2.
The overlay graphs~$G_i$ used at this stage have vertex degrees at most $2^{i+1}\cdot 10$.
By Lemma~\ref{lem:part-two-SCV}, nodes send at most these many messages in Part~2:
\[ 
\sum_{i=1}^{\lceil\lg(t+1)\rceil} \frac{t+1}{2^i} \cdot 
2^{i+1}\cdot 10
\le 
20(t+1)\lceil\lg(t+1)\rceil= \cO(t \log t)
\ .
\]
All these messages are of constant size.
\end{proof}
}

\subsection{Consensus for Few Crashes}

\label{sec:consensus-few-crashes}

We present an algorithm for Consensus that works for $t<\frac{n}{5}$.
The algorithm is called \textsc{Few-Crashes-Consensus}, its pseudocode is in Figure~\ref{fig:few-crashes-consensus}.
It is structured as executing two algorithms: first \textsc{Almost-Everywhere-Agreement}, followed by  \textsc{Spread-Common-Value}.
The first algorithm solves the $\frac{3}{5}$-AEA problem, and the decision obtained in at least $\frac{3}{5}n$ nodes is treated as a common value.
The second algorithm solves the $\frac{3}{5}$-SCV problem, and the common value is used as the ultimate decision.


\begin{figure}[t!]

\hrule

\FF

\texttt{Algorithm} \textsc{Few-Crashes-Consensus}

\FF

\hrule

\FF

\begin{description}[nosep,leftmargin=1em]
\item[]
\texttt{Execute} \textsc{Almost-Everywhere-Agreement}
\item[]
\texttt{if} $p$ decided \texttt{then} adopt the decision as a common value
\item[]
\texttt{Execute} \textsc{Spread-Common-Value} 
\item[]
decide on the common value
\end{description}

\FF

\hrule

\FF

\caption{\label{fig:few-crashes-consensus}
A pseudocode for a node~$p$.   
}
\end{figure}

\begin{theorem}
\label{thm:fom-consensus}

If $t<\frac{n}{5}$ then algorithm \textsc{Few-Crashes-Consensus} solves Consensus in $\cO(t+\log n)$ rounds and  nodes send $\cO(n+t\log t)$ one-bit messages.

\end{theorem}

\begin{proof}
Algorithm \textsc{Almost-Everywhere-Agreement} makes at least $\frac{3}{5}n$ nodes to decide.
Algorithm \textsc{Spread-Common-Value} takes the decision as a common value located in at least $\frac{3}{5}n$ nodes and spreads it across all the nodes.
The agreement, validity and termination properties, as well as the performance bounds, follow directly from Theorems~\ref{thm:AEA} and~\ref{thm:SCV}.
\end{proof}

\subsection{Consensus for Many Crashes}
\label{sec:consensus-many-crashes}

We give a Consensus algorithm that works for a bound $t$ on the number of crashes as large as $n-1$.
The algorithm is called \textsc{Many-Crashes-Consensus}($\alpha$), its pseudocode is in Figure~\ref{fig:fmo-consensus}.
We use the notation~$\alpha=\frac{t}{n}$.
The algorithm has a structure similar to algorithm \textsc{Few-Crashes-Consensus} in Section~\ref{sec:consensus-few-crashes}, in the it consists of broadcasting input value, local probing, and inquiring about decision.
These operations make three parts of an execution.
The first two parts provide almost everywhere agreement and the third part spreads the decision as a common value.
Each part, as well as each phase in the third part, employs its assigned overlay graph.
This is a Ramanujan graph, denoted~$G$, in each of the first two parts.
In the third part, each phase~$i$ uses an overlay graph~$G_{i}$ from a family of $\cO(\lg n)$ graphs. 
Edges in the overlay graphs determine pairs of nodes that communicate directly.
All these graphs are known by every node, as they are determined by the numbers~$n$ and~$t$.
We specify the overlay graphs next.


\begin{figure}[t!]

\hrule

\FF

\texttt{Algorithm} \textsc{Many-Crashes-Consensus}

\FF

\hrule

\FF

\begin{description}[nosep,leftmargin=2em]
\item[\sf Part~1: Broadcasting:] \

\texttt{for} $n-1$ rounds \texttt{do}
\begin{enumerate}[nosep,leftmargin=2em]
\item[]
\texttt{if} $p$'s initial value is $1$ and it is the first round \texttt{or} 
\begin{enumerate}[nosep,leftmargin=2em]
\item[]
$p$ received $1$ in the previous  round for the first time \texttt{then}
\begin{enumerate}[nosep,leftmargin=2em]
\item[]
set the rumor to $1$ and send it to each of $p$'s neighbors in graph $G$
\end{enumerate}
\end{enumerate}
\end{enumerate}

\item[\sf Part~2: Local probing:] \ 

participate in local probing on graph $G$ during $2+\lg n$ rounds:
\begin{enumerate}[nosep,leftmargin=2em]
\item[]
if either the initial value is $1$ or rumor $1$ received in Part~1 then use rumor $1$ else rumor $0$ 
\end{enumerate}
\texttt{if} survived local probing \texttt{then} decide on the rumor value

\item[\sf Part~3: Inquiring about decision:] \ 

\texttt{for} $i=1$ \texttt{to} $1+\lceil \lg \frac{(1+3\alpha) n}{4} \rceil$ \texttt{do} perform Phase~$i$: 
\begin{enumerate}[nosep,leftmargin=2em]
\item[]
at round $2i-1$ of Part~3: 
\begin{enumerate}[nosep,leftmargin=2em]
\item[]
\texttt{if} not decided yet \texttt{then} send an inquiry to each $p$'s neighbor in graph $G_{i}$
\end{enumerate}
\item[]
at round $2i$ of Part~3: \

\begin{enumerate}[nosep,leftmargin=2em]
\item[]
\texttt{if} decided \texttt{and} an inquiry from a neighbor in graph $G_{i}$ received  at round $2i-1$
\begin{enumerate}[nosep,leftmargin=2em]
\item[]
\texttt{then} respond to each inquiring neighbor with $p$'s decision value
\end{enumerate}
\item[]
\texttt{if} a response for own inquiry arrived \texttt{then} adopt the received value as decision 
\end{enumerate}
\end{enumerate}
\end{description}

\FF

\hrule

\FF

\caption{\label{fig:fmo-consensus}
A pseudocode of the algorithm  for a node~$p$.
Parameter $\alpha$ denotes $\frac{t}{n}$.
}
\end{figure}

The overlay graph $G$ used in the first two parts has a vertex for each node of the network.
This graph~$G$ is a $G(n, d)$ Ramanujan graph, where $d = d(\alpha)=\bigl(\frac{4}{1-\alpha}\bigr)^{8}$. 
The graph's degree $d(\alpha)$ determines the following two quantities:  
\[
\ell=\ell(n,d) =  4nd^{-1/8}=(1-\alpha)n \text{ and } \delta =\delta(d)= \frac{1}{2}(d^{7/8} - 2d^{5/8})
\ .
\]
Each graph $G_i$ used in Phase~$i$ of Part~3 is $G(2n,d_i)$ Ramanujan, for  $d_i = \frac{64}{3 (1-\alpha)(1+3\alpha)}\cdot 2^{i}$.
This means that only graphs that are Ramanujan are used to determine overlay networks for algorithm \textsc{Many-Crashes-Consensus}.

Each node maintains a candidate decision value throughout an execution.
At the start, each node sets its candidate decision value to its input value.
If a candidate decision value or a decision value are sent in a message then we refer to the value sent as \emph{rumor}.
At the first round of Part~1, each node with the initial value~$1$ transmits a message with this rumor to all neighbors in graph~$G$.
In the following rounds of Part~1, if a node~$p$ receives a message with a rumor~$1$ while the current candidate decision value is~$0$, then $p$ sets its candidate value to~$1$ and notifies all the neighbors in~$G$ of this change by sending rumors.
This flooding the network with rumor value~$1$ continues for~$n-1$ rounds.
Part~2 consists of local probing in graph~$G$ for $2+\lg n$ rounds.
Candidate decision values are updated similarly as in Part~1, in that if a node~$p$ receives a message with a rumor~$1$ while the current candidate value is~$0$, then $p$ sets its candidate value to~$1$.
Nodes that survive local probing decide on their candidate values.
The last Part~3 is spent by the nodes that have not decided yet inquiring their neighbors about their decisions.
This is partitioned into $1+\lceil \lg \frac{(1+3\alpha) n}{4} \rceil$ phases, with graph~$G_i$ is used in phase~$i$.
If a node that has already decided receives an inquiry about decision then it responds by sending a message with the decision value.
If an inquiry by a node that has not decided yet results in receiving a response that brings a neighbor's decision, then the inquiring node adopts the received decision value as its own decision.

\begin{lemma}
\label{lem:survival-part-two}

At least $\frac{3}{4}(1-\alpha)n$ non-faulty nodes decide on the same value while executing Part~2.
\end{lemma}

\begin{proof}
Let a set $A$ consist of the nodes that are operational at the end of Part~1.
Set $A$ induces a subgraph of $G_A$.
By the end of Part~1, all the nodes in a connected component of $G_A$ set its rumor to the same value.
This is because the length of a longest path is $n-1$ so $n-1$ rounds suffice to propagate $1$ through a connected component.
At most~$\alpha n$ nodes may crash, so there are at least $n - \alpha n=(1-\alpha) n$ non-faulty nodes in~$A$.
Graph $G$ is a $G(n, d)$ Ramanujan graph for $d = \bigl(\frac{4}{1-\alpha}\bigr)^{8}$. 
We defined $\ell =  4nd^{-1/8}$, which equals  $(1-\alpha) n$ by the choice of degree~$d$.
Graph $G$ is $(\ell, \frac{3}{4}, \delta)$-compact, for $\delta = \frac{1}{2}(d^{7/8} - 2d^{5/8})$, by Theorem~\ref{thm:ramanujan-compact}.
It follows, by the definition of compactness, that there exists a set~$S$ that is a $\delta$-survival set with respect to the set of vertices~$A$. 
A $(\gamma(n),\delta(d))$-dense-neighborhood of a vertex includes at least $\ell(n,d)$ vertices, if $\gamma(n) \ge 2+\lg n$ and $d$ is sufficiently large, by Theorem~\ref{thm:sparse-dense}.
Any two such $(\gamma(n),\delta(d))$-dense neighborhoods are connected by an edge, by Theorem~\ref{thm:ramanujan-expanding}.
It follows that nodes that belong to such $(\gamma(n),\delta(d))$-dense neighborhoods acquire equal rumor values by the end of Part~1.
By Proposition~\ref{pro:local-probing}, every node in a $\delta$-survival set~$S$ survives local probing in Part~2. 
Nodes that survive local probing decide, by the pseudocode in Figure~\ref{fig:fmo-consensus}.
There are at least $\frac{3}{4}\ell$  elements in a $\delta$-survival set~$S$, because graph~$G$ is $(\ell, \frac{3}{4}, \delta)$-compact.
The number of non-faulty nodes that survive local probing and decide is at least $\frac{3}{4}\ell=\frac{3}{4}(1-\alpha)n$.
\end{proof}

Let $M=\frac{1}{4}(1+3\alpha) n$, so that Part~$3$ consists of $1+\lceil \lg M \rceil$ phases.

\begin{lemma}
\label{lem:end-part-three}

At most $\frac{M}{2^i}$ non-faulty nodes have not decided yet in Phase~$i$ by the end of round~$2i$ of Part~$3$, for each $i$ such that $0\le i\le 1+\lceil \lg M \rceil$.
\end{lemma}

\begin{proof}
Let a conceptual Phase~$0$ end at the end of Part~2.
Let a set $B$ consist of the nodes that have decided by the end of Part~$2$.
Let $B'\subseteq B$ be a set of non-faulty nodes that have already decided by  the beginning of Part~3. 
We have that $|B'|\ge \frac{3}{4}(1-\alpha) n$, by Lemma~\ref{lem:survival-part-two}.
Let a set $C$ consist of these nodes that have not decided yet by the end of Part~2. 
The complement of set~$B$ has at most~$n-\frac{3}{4}(1-\alpha)n=M$ elements. 
In particular, we have that $|C| \le M$, thus the claim we want to demonstrate holds for~$i = 0$.

Suppose, to the contrary, that this claim does not hold for some $i$ such that $1\le i\le 1+\lceil \lg M \rceil$. 
This means that there exists a set~$C'\subseteq C$ of non-faulty nodes that have not decided yet by round~$2i$, and that the set has more than $\frac{M}{2^{i}}$ elements, for some~$i\ge 1$. 
The overlay graph~$G_{i}$, for $i\ge 1$, used for communication in Phase~$i$ by nodes in $C'$ for sending inquiries and receiving responses is a Ramanujan graph of degree $d_i = \frac{64}{3 (1-\alpha)(1+3\alpha)}\cdot 2^{i}$.
By Theorem~\ref{thm:ramanujan-unbalanced-expansion} applied to the sets $B'$ and $C'$, and constant $\epsilon$ set to $\frac{3}{4}(1-\alpha)$, there is an edge between the set~$B'$ and the set~$C'$. 
This means that the decision was relayed by some element in~$B'$ to some element in~$C'$ in Phase~$i$ via that edge. 
It follows that the nodes in the set~$C'$ decide by the end of Phase~$i$.
This contradicts the definition of~$C'$ and validates of the claim.
\end{proof}

Next, we summarize the performance bounds of algorithm \textsc{Many-Crashes-Consensus}.

\begin{theorem}
\label{thm:Many-Crashes-Consensus}

Algorithm \textsc{Many-Crashes-Consensus} solves Consensus  in  at most $n+3(1+\lg n)$ rounds and nodes send at most $\bigr(\frac{5}{1-\alpha}\bigr)^{8} \,n\lg n$ one-bit  messages, for  sufficiently large~$n$.
\end{theorem}

\begin{proof}
Nodes propagate to their neighbors either their own initial values or what they receive from neighbors.
If a node decides then  it is either on its own initial value or some received one, which, by induction, is some other node's initial value  forwarded during the execution. 
This gives validity.
By Lemma~\ref{lem:end-part-three}, at most $\frac{M}{2^i}$ non-faulty nodes have not decided yet in Phase~$i$ by the end of round~$2i$ of Part~$3$, for $i$ such that $0\le i\le 1+\lceil \lg M \rceil$.
The number $\frac{M}{2^{1+\lceil \lg M \rceil}}$ is less than $1$, which means that the set of nodes that have not decided at round $1+\lceil \lg M \rceil$ of Part~$3$ is empty.
This gives termination.
The decisions made in Part~3 are made after learning of decisions of other nodes during Part~2, so they all are equal, by Lemma~\ref{lem:survival-part-two}.
This gives agreement and completes showing correctness of the algorithm.

Next, we estimate the performance measures: the number of rounds by halting and the number of bits sent in messages.
The pseudocode in Figure~\ref{fig:fmo-consensus} determines the number of rounds.
Part~$1$ takes $n-1$ rounds, Part~$2$ takes $2+\lg n$ rounds, and Part~$3$ takes at most these many rounds: $2(1+\lceil \lg \frac{(1+3\alpha) n}{4} \rceil)\le 2 -4 +2\lg(1+3\alpha) + 2\lg n\le 2+2\lg n$.
All together this makes $n-1 +4+3\lg n= n+3(1+\lg n)$ rounds.

Now we estimate the number of messages. 
In the first two parts, the nodes use graph $G$ in which each vertex has degree $d = \bigl(\frac{4}{1-\alpha}\bigr)^{8}$.
In Part~1, a node sends at most $d$ messages, because each node sends a message to each neighbor at most once.
In Part~2, a node sends at most $d(2+\lg n) $ messages.
In Part~3, nodes use overlay graphs $G_i$ of degree~$d_i$. 
All nodes send at most these many messages in Part~3:
\begin{gather*}
4\sum_{i=1}^{1+\lceil \lg M \rceil} \frac{M}{2^i} \cdot d_i 
= 
\sum_{i=1}^{1+\lceil \lg M \rceil} \frac{M}{2^i} \cdot \frac{2^8}{3 (1-\alpha)(1+3\alpha)}\cdot 2^{i} 
\le
 \frac{2^8 M}{3 (1-\alpha)(1+3\alpha)} (2+\lg M) 
 \le
 \frac{2^6 n\lg n}{3 (1-\alpha)} 
\ ,
\end{gather*}
because $M=\frac{1}{4}(1+3\alpha) n$.
Summing up the bounds for the three  parts, we obtain 
\[
nd (3+\lg n) + \frac{2^8}{3 (1-\alpha)} n\lg n
=
\frac{4^8}{(1-\alpha)^8} \,n\lg n +\frac{2^8 }{3 (1-\alpha)}\, n\lg n +\cO(n) 
\le \frac{5^8 }{(1-\alpha)^8} \,n\lg n
\]
as a bound on the number of messages, for sufficiently large~$n$.
\end{proof}

The greatest possible $t$ is $n-1$ which translates into $\alpha=1-\frac{1}{n}$.

\begin{corollary}

Algorithm \textsc{Many-Crashes-Consensus} solves binary consensus with up to $t=n-1$ crashes in  at most $n+3(1+\lg n)$ rounds such that at most $5^{8} n^9\lg n$ one-bit  messages are sent in total, for  sufficiently large~$n$.
\end{corollary}

\begin{proof}
If $t=n-1$ then $\frac{1}{1-\alpha}=\frac{1}{1-\frac{n-1}{n}}=n$.
The bound follows from Theorem~\ref{thm:Many-Crashes-Consensus}.
\end{proof}

\section{Gossiping with Crashes}
\label{sec:gossip}


We develop a gossiping algorithm \textsc{Gossip}, whose pseudocode is in Figure~\ref{fig:gossip}.
It is designed to work with the assumption $t<\frac{n}{5}$, so that there are $5t$ little nodes.

The little nodes iterate local probing on a Ramanujan graph~$G(5t, d)$ for $d = 5^{8}$.

The surviving nodes first absorb input values in Part~1 and then propagate them in Part~2 using expanding graphs~$G_i$ of increasing degrees, the same as in Lemma~\ref{lem:graph-G-i}.

A node~$p$ represents its extant set as a collection of pairs $(q,x)$, for all node names $q$, where  $x$ is the $q$'s rumor, if $p$ learned that rumor, or $x$ is \texttt{nil} otherwise.
A pair $(q,x)$ is \emph{proper for~$q$} if $x$ is a rumor, and so not \texttt{nil}, otherwise it is \emph{nil for~$q$}.
A node $q$ is \emph{present at~$p$} if the $p$'s  extant set includes a proper pair for~$q$ and it is \emph{absent at~$p$} if the $p$'s  extant set includes the nil pair for~$q$.
A node~$p$ \emph{updates~$q$ to~$(q,x)$} if $p$ replaces a nil pair for~$q$  by a proper pair~$(q,x)$ in the extant set.

Each node maintains a set of nodes called \emph{completion set}.
This set is initialized to the singleton set~$\{p\}$ at a node~$p$.
A node $p$ that receives a completion set from another node $q$ \emph{updates} its completion set by adding to its completion set all the names of nodes in the $q$'s completion set that are absent in the $p$'s completion set.

\begin{figure}[t!]

\hrule

\FF

\texttt{Algorithm} \textsc{Gossip}

\FF

\hrule

\FF

\begin{description}[nosep,leftmargin=2em]

\item[\sf Part~1: Build extant sets:] \ 

initialize the extant set to include a nil pair for each node, and update $p$ to $p$'s pair

\texttt{for} $i=1$ \texttt{to} $\lceil \lg n \rceil$ \texttt{do} Phase $i$ 
\begin{enumerate}[nosep,leftmargin=2em]
\item[]
at round $1$ of Phase $i$: 
\begin{enumerate}[nosep,leftmargin=2em]
\item[]
\texttt{if} $p$ is little \texttt{and} (\texttt{if} $i>1$ \texttt{then} $p$ survived local probing in Phase $i-1$) \texttt{then}
\begin{enumerate}
    \item[]
    send inquiry to each neighbor $u$ in graph $G_{i}$ who is absent at $p$
\end{enumerate}
\end{enumerate}
\item[]
at round $2$ of Phase $i$: \

\begin{enumerate}[nosep,leftmargin=2em]
\item[]
\texttt{if} an inquiry from a neighbor in graph $G_{i}$ received at first round \begin{enumerate}[nosep,leftmargin=2em]
\item[]
\texttt{then} respond to each inquiring neighbor with $p$'s pair
\end{enumerate}
\item[]
\texttt{if} a response $(q,x)$ received for own inquiry \texttt{then} update $q$ to $(q,x)$ 
\end{enumerate}
\item[]
during the following $2+\lg (5t)$ rounds of Phase $i$ execute local probing: \ 
\begin{enumerate}
\item[]
\texttt{if} $p$ is little \texttt{then} participate in local probing on graph $G$:
\begin{enumerate}[nosep,leftmargin=2em]
\item[]
(keep sending the current extant set to all neighbors in $G$ and \\
update each absent node $q$ to $(q,y)$ after receiving a proper pair $(q,y)$) 
\end{enumerate}
\texttt{else} stay idle for $2+\lg (5t)$ rounds
\end{enumerate}
\end{enumerate}

\item[\sf Part~2: Build completion sets:] \ 

 initialize the completion set to the singleton set $\{p\}$

\texttt{for} $i=1$ \texttt{to} $\lceil \lg n \rceil$ \texttt{do} Phase $i$: 
\begin{enumerate}[nosep,leftmargin=2em]
\item[]
at round $1$ of Phase $i$: 
\begin{enumerate}[nosep,leftmargin=2em]
\item[]
\texttt{if} $p$ is little \texttt{and}  (\texttt{if} $i>1$ \texttt{then} $p$ survived local probing in Phase $i-1$) \texttt{then}
\begin{enumerate}
    \item[]
    send  the extant set to each neighbor~$q$ in~$G_{i}$ that is not in the completion set \\ and add $q$ to the completion set
\end{enumerate}
\end{enumerate}
\item[]
at round $2$ of Phase $i$: \

\begin{enumerate}[nosep,leftmargin=2em]
\item[]
\texttt{if} an extant set received from a neighbor in graph $G_{i}$ at the first round 
\begin{enumerate}[nosep,leftmargin=2em]
\item[]
\texttt{then} update each absent node $q$ to $(q,y)$ if a proper pair $(q,y)$ for $q$ received 
\end{enumerate}
\end{enumerate}
\item[]
in the next $2+\lg (5t)$ rounds of Phase $i$ execute local probing: \ 
\begin{enumerate}
\item[]
\texttt{if} $p$ is little  \texttt{then} participate in local probing on graph $G$:
\begin{enumerate}[nosep,leftmargin=2em]
\item[]
(keep sending the current completion set to all neighbors in $G$ and \\
update own completion set after receiving a neighbor's completion set) 
\end{enumerate}
\texttt{else} stay idle for $2+\lg (5t)$ rounds
\end{enumerate}
\end{enumerate}

\end{description}

\FF

\hrule

\FF

\caption{\label{fig:gossip}
A pseudocode of the algorithm  for a node~$p$.
Messages are sent only over the links determined by the overlay graphs denoted as $G$ and~$G_i$, specialized to each part and phase.
}
\end{figure}
\begin{lemma}\label{lem:lit-sur}
At the end each phase $i$, $1 \le i \le \ceil{\lg{n}}$, of Part~1, there is a set of $\Theta(n)$ little nodes that survived local probing of this phase.
\end{lemma}
\begin{proof}
Let a set $B$ consist of the non-faulty little nodes in an execution.
Since there are at least $5t$ little nodes, thus $B$ has at least $5t-t=4t$ elements.
Graph $G$ is a $G(5t, d)$ Ramanujan graph for $d = 5^{8}$. 
We defined $\ell =  20 t d^{-1/8}$, which equals  $4t$ by the determination of degree~$d$. Graph $G$ is $(\ell, \frac{3}{4}, \delta)$-compact, for $\delta = \frac{1}{2}(d^{7/8} - 2d^{5/8})$, by Theorem~\ref{thm:ramanujan-compact}.
It follows, by the definition of compactness, that there exists a $\delta$-survival subset of~$B$ of at least $\frac{3}{4}\ell$ elements. 
By Proposition~\ref{pro:local-probing}, every little node in such a $\delta$-survival set survives local probing in each phase of Part~1.
\end{proof}
\begin{lemma}\label{lem:lit-upd}
The little nodes updates one another with a delay of a single phase at most.
\end{lemma}
\begin{proof}
Consider a proper pair $(r, x)$ present in a little node $q$ in the beginning of phase $i$ of Part~1. We will prove that $q$, if survives the local probing of this phase, updates any other little node that survives the same local probing. 
A node that survives local probing belongs to a $(2+\lg (5t),\delta(d))$-dense-neighborhood, by the duration of local probing in Part~2, and the mechanism of local probing determined by~$\delta$.
By Theorem~\ref{thm:sparse-dense}, $(\gamma(5t),\delta(d))$-dense-neighborhood of a vertex includes at least $\ell(5t,d)$ vertices, if $\gamma(5t) \ge 2+\lg (5t)$.
Any two such neighborhoods are connected by an edge, by Theorem~\ref{thm:ramanujan-expanding}.
Since the duration of the local probing is $2(2+\lg(5t))$, these neighborhoods remain non-faulty sufficiently long assuming it centers survives the local probing. It follows that node $p$ and any other node that survives the local probing update one another about all proper pairs they have in the beginning of the local probing. 
\end{proof}
\begin{lemma}\label{lem:nil-extinct}
The number non-faulty nodes that are nil to a survived little node is at most $\frac{n}{2^{i}}$, by the end of phase $i$ of Part~1.
\end{lemma}
\begin{proof}
By expansion of $G_i$, the set of nodes not known (in terms of extant sets) by the survival set is shrinking exponentially in line with expansion and becomes empty after $G_{\log n}$ is used. 
\remove{Suppose, to the contrary, that the claim of current lemma does not hold for some $i$ such that $1\le i\le \lceil \lg n \rceil$. 
This means that there exists a set~$C_i$ of non-faulty nodes that have not communicated with a little node in round~1 of Phase~$i$, such that $|C_i|=\frac{t}{2^{i+1}}$, for some~$i\ge 1$. 
The overlay graph~$G_{i}$, for $i\ge 1$, used in Phase~$i$ by nodes in $C_i$ for sending inquiries and receiving responses is a graph of vertex expansion $\left(1-\epsilon\right)2^{i}\cdot 10$. We also have that $|C_{i}| = \frac{t}{2^{i + 1}} \le \frac{n}{\left(1 - \alpha\right)\left(1-\epsilon\right)2^{i}\cdot 10}$ since $t < n/5$. 

???By the vertex expansion property applied to set $C_i$, we conclude that this set has at least $\left(1-\epsilon\right)2^{i}\cdot 10 |C_{i}| = \left(1-\epsilon\right) 5t$ neighbors in graph $G_{i}$. 
By Lemma~\ref{lem:part-one-SCV}, at most $2t$ nodes do not decide by the end of Part~1. Adding additional slack of $t$ for the number of faulty nodes, still allows us to reason that at least one neighbor of $C_{i}$ has to decide before the beginning of Part~2, and therefore, the decision was relayed by that 
process to some of its neighbor processes in~$C_i$ during Phase~$i$. 
It follows that some process in $C_i$ decides 
by the end of Phase~$i$, which is a contradiction with the definition of~$C_i$ and thus validates the claim of the lemma.Incorrect??? }
\end{proof}

\begin{theorem}
\label{thm:gossip}
Assuming $t<\frac{n}{5}$, algorithm \textsc{Gossip} solves gossiping in $\cO(\log n \log t)$ rounds with nodes sending  $\cO(n+t\log n \log t)$  messages.
\end{theorem}

\begin{proof}
It follows from Lemmas~\ref{lem:lit-sur},~\ref{lem:lit-upd},~\ref{lem:nil-extinct} that a proper pair of every non-faulty node belongs to the extant set of every survived little node after Part~1 finishes. The same set of Lemmas can be used to show that the extant set, completed after Part~1, are propagated to all non-faulty nodes in Part~2. Thus correctness follows.

The running time performance follows directly from the algorithm's structure represented in the pseudocode: each of the two parts has $\cO(\log n)$ phases, each of $\cO(1+\log t)$ rounds.

To assess the amount of communication: a total of $\cO(t\cdot \log n \log t)$ messages are sent in local probing instances on graph $G$ and $\cO(n+t\log n)$ messages are sent via edges of graph~$G_i$, by arguments similar to those in the proof of Theorem~\ref{thm:fom-consensus}
\end{proof}

\section{Checkpointing with Crashes}

\label{sec:checkpoining}

We present an algorithm for checkpointing with faults modeled as crashes.
The algorithm consists of two parts.
In the  first part, the nodes gossip their names by executing algorithm \textsc{Gossiping} in which each node has the same dummy rumor.
Gossiping ends with each node producing an extant set of nodes, while different nodes may have different sets.
In the second part, the nodes execute $n$ concurrent instances of consensus, each implemented by algorithm \textsc{Few-Crashes-Consensus}.
At a node $p$, the $i$-th instance has input $1$ if node~$i$ is present at $p$, in the sense of the extant set produced by gossiping,  and with input~$0$ otherwise. 
A node transmits messages over a link simultaneously for each instance of consensus, and these messages are combined into one big message to be transmitted together in one round.
A node $p$ decides on an extant set consisting of the names of nodes~$i$ such that decision is on $1$ in the $i$-th instance of consensus.
The algorithm is called \textsc{Checkpointing}, its pseudocode is in Figure~\ref{fig:checkpointing}.

\begin{figure}[t]

\hrule

\FF

\texttt{Algorithm} \textsc{Checkpointing}

\FF

\hrule

\FF

\begin{description}[nosep,leftmargin=2em]

\item[\sf Part~1: Gathering extant sets] \ 

\begin{enumerate}
    \item[]
execute {\sc Gossip} with each node having a dummy rumor 
\end{enumerate}

\item[\sf Part~2: Concurrent instances of consensus] \ 

\begin{enumerate}
    \item[]
execute \textsc{Few-Crashes-Consensus} on $n$ concurrent instances of  consensus:
\begin{enumerate}
    \item[]
the $i$-th instance with input $1$ if node~$i$ is present at $p$  and with input~$0$ otherwise
\end{enumerate}
\end{enumerate}
\item[\sf Decide:] on a final extant set consisting of nodes $i$ such that decision is on $1$ in the $i$-th instance of consensus
\end{description}

\FF

\hrule

\FF

\caption{\label{fig:checkpointing}
A pseudocode of the algorithm  for a node~$p$ for the multi-port model.
A message in Part~2 combines individual messages from all concurrent instances of consensus.
}
\end{figure}

\begin{theorem}
\label{thm:checkpointing}

Assuming $t<\frac{n}{5}$, algorithm {\sc Checkpointing} performs checkpointing in $\cO(t+\log n \log t)$ rounds with nodes sending $\cO(n+t\log n \log t)$ messages.
\end{theorem}

\begin{proof}
All the nodes decide on the same final extant set of names, since it is specified by first agreeing on each member independently.
Running time and communication performance follow from the respective performance bounds of algorithms \textsc{Few-Crashes-Consensus} and  \textsc{Gossip}, as summarized in Theorems~\ref{thm:fom-consensus} and~\ref{thm:gossip}. 
\end{proof}


\section{Byzantine Faults with Authentication}

\label{sec:ab-consensus}



Dolev and Strong~\cite{DolevS83} proposed a deterministic algorithm for 
\dk{Byzantine Broadcast} 
in the model of Byzantine faults and with authentication.
We call it \textsc{DS-algorithm} and use as a sub-routine.
\dk{In this algorithm, there is a source with initial message, and the goal is that all non-faulty processes get the message authenticated by a majority of valid signatures; if the source is Byzantine, either all non-faulty nodes have its initial message or all have value null authenticated.}

Our algorithm is structured into 
\dk{four}
parts.
In Part 1, nodes with the smallest  
\dk{$5t$}
node names, \dk{called little nodes,} initiate \textsc{DS-Algorithm}, which guarantees they accomplish 
\dk{Byzantine Broadcast of each little node's value. 
These are, in fact, parallel $5t$ executions of \textsc{DS-Algorithm} run by the same $5t$ little nodes -- local computations could be done separately at each node, while messages could be combined so that in each round and each pair of sender-receiver, at most one combined message is sent.
At the end, every little node has the same set of initial values (some of them could be null values), although the sets of authenticating signatures could differ. We call such sets {\em authenticated common sets of values}.}


\begin{figure}[t!]

\hrule

\FF

\texttt{Algorithm} \textsc{AB-Consensus} 

\FF

\hrule

\FF

\begin{description}[nosep,leftmargin=2em]
\item[\sf Part~1: \dk{Little nodes obtain authenticated common sets of values:}] \

\texttt{if} 
\dk{$p$ is little}
\texttt{then} 
\begin{enumerate}[nosep,leftmargin=2em]
\item[] execute \textsc{DS-Algorithm}  in $t+1$ rounds, \dk{combining messages with same sender and receiver generated by different parallel executions of \textsc{DS-Algorithm}}
\end{enumerate}
\texttt{else} stay idle for 
\dk{$t+1$}
rounds
\dk{
\item[\sf Part~2: Notifying related nodes of  authenticated common set of values:] \ 

\texttt{if} 
$p$ is little
\texttt{then} notify each related node of the  authenticated common set of values 

\texttt{if} $p$ received an authenticated common set of values from its related little node 
\texttt{then} 
adopt the set

\item[\sf Part~3: Slow propagation of authenticated common set of values:] \ 

\texttt{if} $p$ 
has an authenticated common set of values
\texttt{then} send an authenticated common set of values (arbitrarily chosen, if it has more than one) to all $p$'s  neighbors in graph~$H$

\texttt{for} $\lceil \log_{\frac{3}{2}} \frac{2n/5}{\max\{t,n/t\}} \rceil $ rounds \texttt{do}
\begin{enumerate}[nosep,leftmargin=1em]
\item[]
\texttt{if} $p$ 
has not possessed any authenticated common set of values yet \texttt{and} 
received an authenticated common set of values in the previous round \texttt{then}

\begin{enumerate}[nosep,leftmargin=1em]
\item[]
adopt a received authenticated common set of values (arbitrarily chosen, if it received more than one such set) \texttt{and} forward it to all $p$'s neighbors in $H$ 
\end{enumerate}
\end{enumerate}

\item[\sf Part~4: Fast propagation of authenticated common set of values:] \


\texttt{if} 
$p$ does not have any authenticated common set of values yet
\texttt{then} send an authenticated inquiry to every little node

\texttt{if} $p$ is little 
\texttt{and} 
an authenticated inquiry  received  \texttt{then} 
respond to each inquiring node with~an 
authenticated common set of values

\texttt{if} an authenticated common set of values received as a response for $p$'s own inquiry \texttt{then} adopt the set

\texttt{decide} on the maximum value in the possessed authenticated common set of values
}



\end{description}

\FF

\hrule

\FF

\caption{\label{fig:auth-consensus}
A pseudocode  for a node~$p$.
We assume an authentication mechanism. 
The parameter~$t$ satisfies $0\le t<\frac{n}{2}$ and represents an upper bound on the number of Byzantine nodes.
\dk{Little node is a node with name at most $5t$, and its related nodes are those with the same remainder modulo $5t$. An authenticated common set of values denotes a set of initial messages of little nodes, each authenticated by (at least $4t$) little nodes' valid signatures (in case of faulty initiator, it could be a null value with same type of authentication). Authenticated inquiry means an inquiry message signed by a valid node's signature. Each node can verify if a received set is an authenticated common set of values and if an inquiry is valid -- if not, they are being dropped.}
}
\end{figure}

\dk{In Part~2, little nodes 
send their authenticated common sets of values to their related nodes -- in this section, a related node is a node with the same remainder modulo 
$5t$,
which is the size of the group of the little nodes.}

\dk{In Part~3, nodes propagate their authenticated common sets of values 
using the same mechanism as in Part~1 of algorithm \textsc{Spread-Common-Value}. 
Every time a node receives a set of values of little nodes, it verifies if it contains all values of little nodes and each such value has at least $4t$ valid signatures of little nodes. If not, the node skips such set, and if two such sets or more arrive, one is chosen arbitrarily.}

\dk{In Part~4, nodes that do not have any authenticated common set of values, send authenticated inquiries (by its valid signature) to every little node, which reply with such set after verifying authenticity of the signature (provided that they are not Byzantine).}

\dk{At the end, all non-Byzantine
derive their decisions from their authenticated common set of values in the same manner as any other non-faulty process derives a decision at the completion of an execution of \textsc{DS-Algorithm}, for instance, on the maximum value.}

The algorithm is called \textsc{AB-Consensus}, its pseudocode is in Figure~\ref{fig:auth-consensus}.

\begin{theorem}
\label{thm:ab-consensus}

Assuming Byzantine faults, authentication and $t<\frac{n}{2}$, algorithm \textsc{AB-Consensus} solves consensus with  up to~$t$ Byzantine faults  in $\cO(t)$ rounds with non-faulty nodes sending 
$\cO(t^2+n)$ 
messages. 
\end{theorem}

\begin{proof}
The proof is similar to that of Theorem~\ref{thm:fom-consensus}. We first argue that all non-Byzantine nodes decide on the same value among their initial values.

\dk{%
Algorithm solving almost everywhere agreement in algorithm {\sc Few-Crashes-Consensus} is replaced by Parts~1 and~2: parallel execution of $5t$ instances of \textsc{DS-Algorithm}, followed by sending authenticated common sets of values to related nodes.
This way, all non-faulty little nodes and their related non-faulty nodes have the same authenticated common sets of values (again, the sets of signatures could differ, but the values associated with little nodes are the same). All these values are initial values, could also be null in case of Byzantine originator -- but at least $4t$ are non-Byzantine. It follows directly from~\cite{DolevS83} and authentication mechanism used by little nodes and their related nodes. 

The same type of argument as in Lemma~\ref{lem:part-one-SCV}
combined with authentication mechanism allowing to recognize whether a received set is an authenticated common set of values, at the end of Part $3$ at most $\max\{Ct, t + \frac{n}{t}\}$ non-faulty nodes do not have any authenticated common set.

Finally, in Part 4, all such nodes generate authenticated inquiries and send them to all little nodes. The little nodes verify the signatures, and send their authenticated common set of values as a reply. Note that each little node can receive at most $\max\{Ct, t + \frac{n}{t}\}$ such requests from non-faulty nodes and at most $t$ from Byzantine nodes (because of authentication, Byzantine nodes cannot forge messages, pretend to be other nodes and thus generate more than one authenticated message for any little node). This way all non-faulty nodes get an authenticated common set, as there is at least one non-faulty little node that has such set. At the end of Part~4, each node makes decision on the maximum value in its authenticated common set, and since they are the same sets of values -- it is the same and is one of the initial values.

We now argue about time and message complexities.
Part~1 works in running time $\cO(t)$, by~\cite{DolevS83}, and uses $\cO(t^2)$ messages, as it is executed by $\cO(t)$ nodes only. Part~2 takes only one round
and $O((2t+1)\cdot \frac{n}{2t+1})=O(n)$ point-to-point messages sent by non-faulty nodes. 
Part~3 takes $O(\log t)$ time and each non-faulty node sends $O(1)$ messages, hence $O(n)$ messages in total.
Finally, Part~4 takes only two rounds, during which at most $\max\{Ct, t + \frac{n}{t}\}$ non-faulty nodes send $5t$ inquiring messages, each. The total number of such inquiries is, therefore, $O(t^2+n)$. The number of replies sent by non-faulty processes could be a bit bigger, as little nodes may reply to authenticated inquiries coming from Byzantine nodes -- there are, however, at most $t$ such inquiries arriving at each non-faulty little node, hence the additive overhead on the number of replies is only $O(t^2)$. Thus, the total time complexity is $O(t)$ and total number of point-to-point messages set by non-faulty processes in $O(t^2+n)$.
%
}%
\end{proof}


\section{The Single-Port Model}

\label{sec:single-port}

In this section, we explain how the consensus algorithms for crashes from Section~\ref{sec:agreement-crashes}, designed and analyzed for the multi-port model, can be adapted to the single-port model while maintaining the optimality of asymptotic performance bounds. 
The consensus algorithm adapted to the single-port model is called \textsc{Linear-Consensus}.
It can be implemented in the single-port model with the same asymptotic performance as given in Theorem~\ref{thm:fom-consensus} for the multi-port model, with an extra additive $\cO(\log n)$ term in running time performance bound, as specified in Theorem~\ref{thm:linear-consensus-single-port}. 

We design an adaptation by structuring communication  in the single-port model into  message-passing rounds, called {\em mp-rounds}. 
In order to distinguish  mp-rounds from rounds in the original execution of algorithm \textsc{Linear-Consensus} in the single-port  model, we call the 
latter  {\em sp-rounds}.
Local computation during an mp-round, of algorithm \textsc{Few-Crashes-Consensus}, is performed at the very end of the last sp-round of the implementation of communication part of the mp-round.

Next we discuss the adaptations of parts of algorithms:

Regarding Parts 1 and 2 of \textsc{Almost-Everywhere-Agreement}:
The overlay graph $G$ has constant degree, therefore each original mp-round is implemented in $2d$ sp-rounds in the single-port model, where $d=\cO(1)$ is the degree of~$G$.
In the first $d$ sp-rounds of the implementation of an mp-round, a node that intends to transmit in the mp-round chooses its neighbors in $G$ in an arbitrary order and sends messages to them one-by-one. 
In the remaining $d$ sp-rounds, each nodes goes through its $d$ in-coming ports one-by-one and retrieves messages from them, at most one message per port.

Regarding Part 3 of \textsc{Almost-Everywhere-Agreement} and Parts~1 and~2 of \textsc{Spread-Common-Value}: They require some modifications in the original algorithm. 
Due to the lower bound on consensus runtime in the single-port model, see  Theorem~\ref{thm:sp-lower}, we could add components of time $\Theta(\log n)$.

Regarding Part~3 of \textsc{Almost-Everywhere-Agreement}:
If $t> \sqrt{n}$, then the number of links used by little nodes is $\frac{n}{t}<t$,  so we can schedule communication through these links similarly as in previous parts, which takes time $\cO(t)$ (relative nodes have just one in-coming port to check each round). 
If $t\le \sqrt{n}$, we could apply the same approach as in~Part~1 of \textsc{Spread-Common-Value}, as explained next.

Regarding Part~1 of \textsc{Spread-Common-Value}:
Similarly as in Parts~1 and~2 of \textsc{Almost-Everywhere-Agreement}, an overlay graph~$H$ has constant degree, therefore the same implementation as in Parts~1 and~2 of \textsc{Almost-Everywhere-Agreement} works, number~$d$ being adjusted to the degree~of~$H$. 


Regarding Part 2 of \textsc{Spread-Common-Value}:
At the start of that part the number of nodes that have not decided yet 
is at most $t+1$ in the implementation in the single-port model.
Similarly as in the previous Part~1, the inquiries about decision values in consecutive phases are implemented through inquiring the neighbors in the overlay graph $G_i$ determined by the phase number~$i$. 
It suffices for each node to inquire $3t+1$ nodes to find one with decision, therefore it is enough to schedule phases until the degree of the graph $G_i$ grows beyond $3t$. 
This way each node schedules only $\cO(t)$ links for communication, in the whole Part~2, which gives a running-time bound $\cO(t)$. 
The upper bound on the number of messages remains the same as in the proof of Theorem~\ref{thm:fom-consensus} for the original Part~2, as a subset of original links is used in the single-port implementation.

\begin{theorem}
\label{thm:linear-consensus-single-port}

Algorithm \textsc{Linear-Consensus} can be implemented in the single-port model such that it solves consensus in $\cO(t+\log n)$ rounds with $\cO(n+t\log n)$ bits sent in messages.
\end{theorem}

\begin{proof}
Using arguments formulated above for each analyzed part, an inductive argument over all parts and phases of algorithm \textsc{Linear-Consensus} shows that similar properties as in the original analysis of the counterpart multi-port algorithm \textsc{Few-Crashes-Consensus} hold; with the exception  of the implemented \dk{Part~3 of \textsc{Almost-Everywhere-Agreement} and Part~2 of \textsc{Spread-Common-Value}, which are of length $\cO(t)$ in the single-port model, while Part~1 of \textsc{Spread-Common-Value}, which is of length $\cO(\log n)$ in the single-port model}.
The correctness and time/communication complexity of the \textsc{Linear-Consensus} algorithm  implemented in the single-port model follows from the correctness and complexity of the original algorithm in the multi-port message-passing model, as in Theorem~\ref{thm:fom-consensus}, with an additional additive component $\cO(\log n)$ to the time complexity. 
\end{proof}

We prove a lower bound $\Omega(t+\log n)$ on the running time for the problems of consensus, gossiping and checkpointing in the single-port model. The running time of the consensus algorithm adapted to the single-port model matches this lower bound.

\begin{theorem}
\label{thm:sp-lower}

Any deterministic algorithm solving fault-tolerant consensus, gossip or checkpointing in the single-port model requires time $\Omega(t+\log n)$, for any $0<t<n$.
\end{theorem}

\begin{proof} 
The lower bound $\Omega(t)$ on running time for consensus and checkpointing holds in the general message-passing model, see~\cite{Attiya-Ellen-book-2014,Attiya-Welch-book2004}.
In order to prove the $\Omega(t)$ bound for gossiping, we use the properties of the single-port model.
Consider a deterministic algorithm solving gossiping for $t=0$ crashes; if $t>0$, the adversary may mimic the strategy for $t=0$. 
The adversary may choose a single node arbitrarily, say $v$, and assign to it an arbitrary value, say $0$. 
Consider the following two initial configurations of input values: in $C_0$ all nodes have input value~$0$ and in~$C_1$ all nodes but~$v$ have input value~$1$.

We construct two executions, $\cE_0$ and $\cE_1$ of a gossiping algorithm, round by round, starting on configurations $C_0$ and $C_1$, respectively. 
Let $\cE_0[i],\cE_1[i]$ denote the executions $\cE_0,\cE_1$, respectively, built up to the end of round $i>0$. 
We may assume, for technical reason, that $\cE_0[0]$ and $\cE_1[0]$ are simply the initial configurations $C_0$ and $C_1$, respectively.
Note that the algorithm at node $v$ cannot stop before the first round, as it would have to output different sets of values in the two executions while having same state in both, which would be a contradiction.
In each round the adversary crashes at most two nodes.
Suppose we build, recursively, a sequence of executions up to the end of rounds $0\le i\le \frac{t}{2}$, which are executions $\cE_0[i]$ and $\cE_1[i]$. 
In order to extend the execution to the next round~$i+1$, the adversary runs a simulation of the algorithm and pre-computes which port $v$ is going to use in round $i+1$, based on its states in already built executions $\cE_0[i]$ and $\cE_1[i]$. 
In each of these two executions, the adversary may choose at most one port in round $i+1$, say, to nodes $w_0,w_1$ respectively. 
The adversary builds valid executions $\cE_0[i+1]$ and $\cE_1[i+1]$ as follows: it takes $\cE_0[i]$ and $\cE_1[i]$ and in each of them it additionally crashes nodes $w_0$ and $w_1$ at the start of round~$1$. 

We obtain the following invariant, which follows by induction on $i$: by round $i+1\le $ of executions $\cE_0[i+1]$ and $\cE_1[i+1]$, there are at most $2(i+1)\le t$ nodes crashed and no message is successfully transferred to/from node $v$.

By this invariant, node $v$ is in the same state at the end of these two executions.  
Hence, it cannot halt, because in both these executions it would output the same set of values, implied by its state, which would be a contradiction since the sets of values $C_0$ and $C_1$ in these two executions are  different. This proves the lower bound $\Omega(t)$ for gossiping.

The lower bound $\Omega(\log n)$ on time is first proved for consensus and later extend to gossiping and checkpointing.
We proceed as follows. 
Consider a deterministic algorithm solving consensus for $t=0$ crashes (if $t>0$, the adversary mimic the strategy for $t=0$).
Let an initial configuration be a 0-1 vector that on its $i$-th coordinate has the initial value of the $i$-th node.

We claim that there are two initial 0-1 configurations $C_0$ and $C_1$ that differ only on one coordinate/node $v$ and such that the consensus algorithm agrees on $0$, or~$1$, when starting with $C_0$, or~$C_1$, respectively.

The proof of the claim follows directly from the following observation. 
Consider initial configurations $C^*_{< i}$, defined for $1\le i \le n+1$ in such a way that in $C^*_{<i}$ all nodes with name smaller than~$i$ start with value $0$ and the remaining ones start with value $1$. 
Clearly, in $C^*_{<1}$ all nodes start with value $1$, therefore, by validity, the agreement must be on value $1$. 
Analogously, in $C^*_{<n+1}$ all nodes start with value $0$, therefore, by validity, the agreement must be on value $0$. 
Therefore, there is $1\le i\le n$ such that $C^*_{<i}$ results in agreement on value $1$ while $C^*_{<i+1}$ results in agreement value $0$. 
Note that these two configurations differ on values of exactly one node, the $i$-th one. 
We set $C_0$ to $C^*_{<i+1}$ and $C_1$ to $C^*_{<i}$. 
This completes the proof of the claim.

In the remainder of the proof, we use similar notation as  for gossiping. Consider executions $\cE_0$ and $\cE_1$ of the consensus algorithm starting on configurations $C_0,C_1$, respectively. 
Let $\cE_0[i]$ and $\cE_1[i]$ denote the executions $\cE_0$ and $\cE_1$, respectively, up to the end of round~$i>0$. 
We may assume that $\cE_0[0]$ and $\cE_1[0]$ are simply the initial configurations $C_0$ and $C_1$, respectively.

The following property is an invariant : after round $0\le i< \log_3 n$, at most $3^i$ nodes has different states at the end of executions $\cE_0[i]$ and~$\cE_1[i]$.

The proof of the invariant is by induction on round number~$i$. 
For the basis step, consider the beginning of the computation in round~$0$.
We have that exactly one node has different states in $\cE_0[i],\cE_1[i]$, by definition of the configurations $C_0$ and $C_1$, so therefore  the invariant holds for $i=0$.
For the inductive step, assume that it holds for some $0\le i< \log_3 n - 1$, and we prove it for $i+1$. 
Let $A[i]$ be the set of nodes with different states at the end of both $\cE_0[i],\cE_1[i]$. 
By the invariant for~$i$, we have $|A[i]|\le 3^i$. 
Each of them may send a message to at most one other node in round $i+1$ in each of the two executions, $\cE_0[i+1],\cE_1[i+1]$. 
Therefore, at most $2\cdot |A[i]|$ nodes outside of $A[i]$ may get a message from some node in $A[i]$ in some of $\cE_0[i+1]$ and $\cE_1[i+1]$, and only those nodes may change to different states in executions $\cE_0[i+1]$ and $\cE_1[i+1]$. 
Hence, at the end of executions $\cE_0[i+1]$ and $\cE_1[i+1]$, at most $|A[i]|+2|A[i]|=3|A[i]|\le 3^{i+1}$ may have different states in them, where the last inequality follows from the invariant for~$i$. 
This completes the proof of the invariant for $i+1$, and thus for all $i<\log_3 n$.

Suppose, to the contrary, that all nodes stop in both execution by round $i<\log_3 n$. 
By the choice of configurations $C_0,C_1$, all nodes decide on~$0$ and all on~$1$, respectively. 
By the invariant, there are $n-3^i>0$ nodes that have the same ending state in $\cE_0[i]$ and $\cE_1[i]$, which contradicts that it has to decide on~$0$ in $\cE_0[i]$ and on~$1$ in~$\cE_1[i]$. 
Thus, the algorithm requires $\Omega(\log n)$ rounds in a worst-case execution.

To extend this lower bound to gossiping and checkpointing, consider scenario with no failures. 
First, we arbitrarily choose a function from all subsets of $n$ nodes to values $0,1$ such that it assigns value $0$ to an empty set and value $1$ to the whole set $[n]$. 
Next, we design a consensus algorithm that applies gossiping or checkpointing algorithm to learn a set of nodes starting with value $1$ (while others have initial value $0$), and then apples the function to compute decision value.
Since we consider failure-free scenario, the results of gossiping and checkpointing are the same at every node, therefore they decide on the same value and satisfy validity (by the assumption on the chosen function). Therefore, if there was a gossiping or checkpointing algorithm with time complexity $o(\log n)$, the above mentioned scenario would result in a consensus algorithm in the same time $o(\log n)$, which would be a contradiction with the lower bound $\Omega(\log n)$ on consensus.
\end{proof}

\section{Discussion}

We developed algorithms that provide linear running time and communication for some ranges of upper bound~$t$ on the number of faults with respect to~$n$, as summarized in Table~\ref{table:summary} in Section~\ref{sec:introduction}.
Extending these ranges closer to~$n$ is an open problem.

Determining optimal performance bounds on running time and communication of deterministic algorithms for checkpointing is an open problem.
The  message complexity of the time-optimal checkpointing algorithm for crashes given in Section~\ref{sec:checkpoining} may miss message optimal performance by a poly-logarithmic factor in the multi-port model.
The time performance of the message-optimal algorithm for checkpointing given by Galil, Mayer, and Yung~\cite{GalilMY95} may miss a time-optimal performance by a polynomial factor in the multi-port model.
We conjecture that the optimal running time and communication performance bounds for checkpointing in the single-port model are comparable asymptotically to those   in the multi-port model.

This work demonstrates that Ramanujan graphs used as overlay networks are conducive to structuring synchronous communication to support time and communication efficient algorithms for consensus with nodes prone to  crashes or Byzantine faults with authentication.
We expect that the newly discovered properties of Ramanujan graphs could be applied to streamline performance of algorithms in synchronous distributed systems and communication networks for problems like gossip, counting, and majority consensus. 

It is also intriguing if Ramanujan graphs used as overlay networks  could enhance  communication performance of algorithms in asynchronous message passing~systems.

\bibliographystyle{plain}

\bibliography{references}

\end{document}